\documentclass[aps,pra,twoside,superscriptaddress,showkeys,nofootinbib]{revtex4}
\usepackage{amsmath,amsfonts,amssymb,amsthm,amsbsy,mathrsfs,dsfont,mathtools}
\usepackage{graphicx,epsfig}
\usepackage{bbold}
\usepackage{enumerate}

\newcommand{\ket}[1]{\mathop{\left|#1\right>}\nolimits}            
\newcommand{\bra}[1]{\mathop{\left<#1\,\right|}\nolimits}         
\newcommand{\kbd}{\mathop{\left|\bullet\right>}\nolimits}            
\newcommand{\bbd}{\mathop{\left<\bullet\,\right|}\nolimits}         

\newcommand{\brk}[2]{\langle #1 | #2 \rangle}
\newcommand{\kbr}[2]{| #1\rangle\!\langle #2 |}

\newcommand{\bu}{\bullet}

\newcommand{\Tr}[1]{\mathop{{\mathrm{Tr}}_{#1}}}            
\newcommand{\sTr}[1]{\mathop{{\mathrm{sTr}}_{#1}}}          

\newcommand{\conj}[1]{\mathop{{\mathrm{conj}}{\left(#1\right)}}}          
\newcommand{\id}{\mathop{{\mathrm{id}}}\nolimits}                         

\newtheorem{thm}{Theorem}
\newtheorem{prop}[thm]{Proposition}
\newtheorem{cor}[thm]{Corollary}
\newtheorem{lem}[thm]{Lemma}
\newtheorem{defi}[thm]{Definition}

\theoremstyle{remark}
\newtheorem*{rem}{\bf Remark}
\newtheorem*{exa}{\bf Example}

\def\openone{\mathbb{1}}
\def\bbC{\mathbb{C}}
\def\bbZ{\mathbb{Z}}
\def\bbR{\mathbb{R}}
\def\bbK{\mathbb{K}}
\def\bbN{\mathbb{N}}

\def\SS{\mathscr{S}}

\newcommand{\nn}{\nonumber}
\def\dg{\dagger}
\def\ddg{\ddagger}
\def\df{\overset{\rm df}{=}}
\def\MM{\M(p|q,\C\La_N)}
\def\toddg{\overset{\ddagger}{\to}}

\def\a{\alpha}
\def\b{\beta}
\def\g{\gamma}
\def\Ga{\Gamma}
\def\d{\delta}
\def\e{\epsilon}

\def\vr{\varrho}
\def\vp{\varphi}

\def\om{\omega}

\def\s{\sigma}

\def\la{\lambda}
\def\La{\Lambda}

\def\vt{\vartheta}

\def\eth{\eta^\#}

\def\M{\mathcal{M}}
\def\G{\mathcal{G}}
\def\C{\mathcal{C}}

\begin{document}

\author{Kamil Br\'adler}
\email{kbradler@ap.smu.ca}
\affiliation{
    Department of Astronomy and Physics,
    Saint Mary's University,
    Halifax, Nova Scotia, B3H 3C3, Canada
    }
\affiliation{
    School of Computer Science,
    McGill University,
    Montreal, Quebec, H3A 2A7, Canada
    }

\title{The theory of superqubits -- supersymmetric qubits}

\begin{abstract}
    Superqubits are the minimal supersymmetric extension of qubits. In this paper we investigate in detail their unusual properties with emphasis on their potential role in (super)quantum information theory and foundations of quantum mechanics. We propose a partial solution to the problem of negative transition probabilities that appear in the theory and has been previously reported elsewhere. The modification does not affect the performance of supersymmetric entangled states in the CHSH game -- superqubits provide resources more nonlocal than it is allowed by ordinary quantum mechanics.
\end{abstract}


\maketitle

It is a widely accepted fact nowadays that quantum mechanics is a qualitatively different theory compared to classical mechanics. Quantum mechanics  provides  resources such as  pure or mixed entangled states  that are impossible to simulate in classical physics. But an interesting point was raised by Popescu and Rohrlich in~\cite{PR} almost two decades ago. They asked why quantum mechanics could not have been even more nonlocal than it actually is. First, they recalled an earlier result by Tsirelson~\cite{tsirelson} who  showed that in a certain version of Bell's inequalities~\cite{bell}, known as the CHSH inequality~\cite{CHSH}, no quantum-mechanical state can violate the inequality more than a maximally entangled state. So their question was: Going beyond Tsirelson's bound, does it clash with other established principles of  physics, namely the impossibility of superluminal communication (the so-called no-signalling condition)? Surprisingly, the answer is no and there exists a gap between Tsirelson's bound and unhealthy theories allowing faster-than-light communication. Henceforth, quantum mechanics is not the maximally nonlocal theory it could have been. That immediately raises another question whether a consistent theory inhabiting the gap could be constructed or even realized in Nature. A considerable effort has been recently spent on investigating the consequences of such a superquantum theory~\cite{masanes,ent,complexity} and a number of results showed that if Tsirelson's bound is crossed certain entropic quantities valid in quantum mechanics would become invalid~\cite{dataproc,IC,MI1,MI2}. This would have significant implications for the fields of quantum information theory and foundations of physics~\cite{complexity_review}.

An important tool in these investigations is a hypothetical resource called a nonlocal box also proposed in~\cite{PR}. It is a superquantum resource performing strictly better than quantum mechanics while still respecting the no-signalling condition.  Nonlocal boxes rule over the whole gap in a sense that their decohered version can approximate any theory between quantum mechanics and the no-signalling world. They are, however, purely mathematical constructs with no links to even hypothetical physical theories.

The object of study in this work is called a superqubit. It was first introduced in~\cite{superqubits} and its role in the CHSH game was investigated for the first time in~\cite{bbd}. Put simply, superqubits are the minimal supersymmetric extension of qubits where the main role is played by the orthosymplectic Lie superalgebra over the reals $osp(1|2)$ (more precisely, one of its  real forms). The $osp(1|2)$ algebra has been extensively studied in the past~\cite{snr1,landimarmo,cohstates,supersphere,supersphere1,SchW} as one of the most important example of a Lie superalgebras~\cite{ritt,superLie,superLie1,manin,varadara,berezin_book,kostant,carmeli}. On the physical side, the main motivation for studying supersymmetry comes from high energy physics where it is the leading candidate for physics beyond Standard Model.  It is nonetheless important to stress that the Lie superalgebra and the derived structures studied in this paper are not directly related to any hypothetical superpartner. Despite of this, there are at least three reasons why it is an interesting problem to study. First, there exists at least two proposals from condensed matter physics where supersymmetry, in particular the family of orthosymplectic Lie superalgebras $osp(p|q)$, plays a key role~\cite{efetov,superspinors,hasebe}. Second, the developed methods and concepts can be eventually used for higher-dimensional Lie superalgebras that are relevant to supersymmetry-based high-energy physics. Finally, if we ignore the question of direct physical relevance, it is an exciting quest to create a {\em synthetic} quantum theory that contains ordinary quantum mechanics and non-trivially extends it at the same time.

The result of~\cite{bbd} suggests that the supersymmetric extension of quantum mechanics based on superqubits may be a candidate for a superquantum theory that lies in the gap between ordinary quantum theory and PR boxes. This is due to the violation of Tsirelson's bound reported there. An unfortunate consequence is the presence of negative transition probabilities. They never appear in the actual calculation leading to the result but one would like to avoid them entirely. Here we propose a  solution to this problem for single superqubits. The same issue in the case of multi-superqubit states remains open and is likely to be resolved in the context of a larger Lie supergroup.

This paper has a multiple purpose: (i) to anchor the notion of superqubits on a firm mathematical footing that has its roots in the theory of Lie superalgebras and related structures, (ii) to systematically develop the rules for calculating with superqubits in the same way quantum information theorists deal with qubits and, (iii) to offer a solution to the problem of occurrence of negative probabilities. This is an issue first encountered in~\cite{bbd} and we offer a partial solution by means of compactification of the superqubit space. As a result, the problem of negative probabilities disappears for single superqubits.  An important consequence is that we violate Tsirelson's bound less ($p^{sqbit}_{win}\simeq 0.8647$) than reported in~\cite{bbd}.
From other results it is worth of mentioning that we have introduced a super Hilbert space on the vector subspace of the supermatrix space $\MM$ together with formalizing the notions of a (super)ket and bra based on superlinear algebra. The supermatrix formalism enables us to easy manipulate and calculate with superqubits.

There are two main sections and two appendices in this paper. Section~\ref{sec:squbits} introduces superqubits in a manner  different from the original article~\cite{superqubits}. It attempts to introduce some additional, perhaps less known or novel, details about the used superstructures in order to formalize the notion of a superqubit from the mathematical point of view. It heavily relies on the definitions and terminology of supermathematics summarized and reviewed in Appendix~\ref{sec:AppBackground}. Section~\ref{sec:CHSH} focuses more on the physical consequences of the developed formalism and we investigate the performance of bipartite superqubit states in the CHSH game. Appendix~\ref{sec:AppSLinearALg} brings detailed calculations of some results in Section~\ref{sec:squbits} based on superlinear algebra whose theory is outlined in Appendix~\ref{sec:AppBackground}. The paper concludes with a number of open questions.

\section{Superqubits}\label{sec:squbits}

The purpose of this section is to formally introduce superqubits, their properties and relation to qubits. Superqubits can be understood as a supersymmetric version of qubits studied in quantum information theory. The discussion here is built on some standard knowledge and results in the theory of Lie superalgebras and superlinear algebra that has been reviewed in reasonable detail in Appendix~\ref{sec:AppBackground}. The notation and terminology used in this section has also been defined there. Here we  gather some useful definitions and prove certain facts about the studied superstructures not found anywhere in the literature known to the author. The immediate starting point is Def.~\ref{def:gradeadjoint}.
\begin{thm}\label{thm:ddgIsGradeAdjoint}
  Let $S\in\MM$ be a supermatrix and let the double dagger map be
  $$
  \ddg\df \#\circ ST,
  $$
  where $ST$ stands for supertranspose of a supermatrix, Eq.~(\ref{eq:STforS}), and the hash map $\#$ is a grade involution from Def~\ref{def:automorph}. Let further $\brk{_-}{_-}:\La^{p|q}\times\La^{p|q}\mapsto\C\La_N$ be a bilinear, non-degenerate form where $\C\La_N$ denotes the complex Grassmann algebra of order $N$ and $\La^{p|q}$ is the $\C\La_N$-bimodule (Def.~\ref{def:supermodule}). Then the  double dagger satisfies the properties of the grade adjoint
  \begin{equation}\label{eq:gradeAdjointwithDDG}
    \brk{Sz}{s}=(-)^{|S||z|}\brk{z}{S^\ddg s}
  \end{equation}
  from  Def.~\ref{def:gradeadjoint} valid for all $s,z\in\La^{p|q}$.
\end{thm}
\begin{rem}
    The elements of $\La^{p|q}$ are represented by supermatrices from the set $\MM$ (see Def.~\ref{def:SupermatrixSet}). The proof is given for  supermatrices where $p=q=1$. It can be handled in full generality for $\MM$ but I believe that the explicit calculation that follows will be more enlightening. Moreover, the generalization for an arbitrary $p$ and $q$ is straightforward.
\end{rem}
\begin{rem}
  Note that we do not require the form to be positive semidefinite or Hermitian. This terminology has not been defined for supermatrices. Also, there is likely more than one candidate for the form with the required properties but we will not need an explicit example for the sake of this Theorem. The restriction to a specific bilinear form will appear after the definition of the $uosp(1|2;\C\La_N)$ algebra when the will be shown to be super Hermitian and positive semidefinite (see the introduction of a super Hilbert space on page~\pageref{subsubsec:superH}).
\end{rem}
\begin{rem}
  Here we start using the super ket and bra notation. For its origin, see the important remark on page~\pageref{rem:superkets}.
\end{rem}
\begin{proof}
  For the sake of the proof we set $p=q=1$ as the simplest non-trivial case. Using the standard basis Eq.~(\ref{eq:freebasis}), a homogeneous
  supermatrix $S\in\M(1|1,\C\La_N)$ will be written as
  \begin{equation}
  S=\begin{pmatrix}
    a & b \\
    c & d \\
  \end{pmatrix},
  \end{equation}
  where $a,b,c,d\in\C\La_N$. In particular, (see the discussion after Eq.~(\ref{eq:ABDCmatrix})) $|S|=0$ if $|a|=|d|=0$ and $|b|=|c|=1$ and $|S|=1$ when the grade of the entries is reversed. Superkets and bras that appear in Eq.~(\ref{eq:gradedDEF}) are column and row super vectors, respectively. In this particular case, if $|z|=0$ the vector takes the following form:
  $$
    \ket{z} =
    \begin{pmatrix}
      x \\
      \eta \\
    \end{pmatrix},
  $$
  where $\eta$ is any odd element of $\C\La_N$ (it can be, without loss of generality, one of the generators of the Grassmann algebra $\C\La_N$) and $x$ is an even Grassmann number. For $|z|=1$ we have
  $$
    \ket{z}  =
     \begin{pmatrix}
     \eta \\
     x \\
     \end{pmatrix}
  $$
  and again $|\eta|=1$ and $|x|=0$. \\
  (i) To prove Eq.~(\ref{eq:gradeAdjointwithDDG}) it is sufficient to verify the following four configurations that occur:
  \begin{itemize}
  \item $|S|=0$ and $|z|=0$.\\
      LHS of~(\ref{eq:gradeAdjointwithDDG}):
      $$
      \begin{pmatrix}
        a & b \\
        c & d \\
      \end{pmatrix}
      \begin{pmatrix}
        x \\
        \eta \\
      \end{pmatrix}
      =
      \begin{pmatrix}
        ax+b\eta \\
        cx+d\eta \\
      \end{pmatrix}
      \toddg
      \begin{pmatrix}
        a^\#x^\#+b^\#\eta^\# ,& c^\#x^\#+d^\#\eta^\# \\
      \end{pmatrix}.
      $$
      RHS of~(\ref{eq:gradeAdjointwithDDG}):
      \
      $$
      \begin{pmatrix}
        x^\# ,&\eta^\# \\
      \end{pmatrix}
      \begin{pmatrix}
        a^\# & c^\# \\
        -b^\# & d^\# \\
      \end{pmatrix}
      =
      \begin{pmatrix}
        x^\#a^\#-\eta^\# b^\# ,& x^\#c^\#+\eta^\#d^\# \\
      \end{pmatrix}
      =
      \begin{pmatrix}
        a^\#x^\#+b^\#\eta^\#  ,& c^\#x^\#+d^\#\eta^\# \\
      \end{pmatrix}.
      $$
      The properties of the grade involution $\#$ were used (Def.~\ref{def:automorph}) as well as the definition of the supertranspose, Eq.~(\ref{eq:STforS}). As we can see, both sides are equal for any choice of $s\in\La^{1|1}$ without an explicit form of the bilinear form. We obtain similar agreement in the remaining three possibilities.
  \item $|S|=0$ and $|z|=1$.\\
      LHS:
      $$
      \begin{pmatrix}
        a & b \\
        c & d \\
      \end{pmatrix}
      \begin{pmatrix}
        \eta \\
        x \\
      \end{pmatrix}
      =
      \begin{pmatrix}
        a\eta+bx \\
        c\eta+dx \\
      \end{pmatrix}
      \toddg
      \begin{pmatrix}
        a^\#\eta^\#+b^\#x^\# ,& -c^\#\eta^\#-d^\#x^\# \\
      \end{pmatrix}.
      $$
      RHS:
      \begin{align*}
        \begin{pmatrix}
        \eta^\# ,&-x^\# \\
      \end{pmatrix}
      \begin{pmatrix}
        a^\# & c^\# \\
        -b^\# & d^\# \\
      \end{pmatrix}
      = & \begin{pmatrix}
        \eta^\# a^\#+ x^\#b^\# ,& \eta^\# c^\#-x^\#d^\# \\
      \end{pmatrix} \\
      = & \begin{pmatrix}
        a^\#\eta^\#+b^\#x^\#  ,& -c^\#\eta^\#-d^\#x^\# \\
      \end{pmatrix}.
      \end{align*}
  \item $|S|=1$ and $|z|=0$.\\
    LHS:
      $$
      \begin{pmatrix}
        a & b \\
        c & d \\
      \end{pmatrix}
      \begin{pmatrix}
        x \\
        \eta \\
      \end{pmatrix}
      =
      \begin{pmatrix}
        ax+b\eta \\
        cx+d\eta \\
      \end{pmatrix}
      \toddg
      \begin{pmatrix}
        a^\#x^\#+b^\#\eta^\# ,& -c^\#x^\#-d^\#\eta^\# \\
      \end{pmatrix}.
      $$
      RHS:
      \begin{align*}
        \begin{pmatrix}
        x^\# ,&\eta^\# \\
      \end{pmatrix}
      \begin{pmatrix}
        a^\# & -c^\# \\
        b^\# & d^\# \\
      \end{pmatrix}
      = & \begin{pmatrix}
        x^\#a^\#+\eta^\# b^\# ,& -x^\#c^\#+\eta^\# d^\# \\
      \end{pmatrix} \\
      = & \begin{pmatrix}
        a^\#x^\#+b^\#\eta^\#  ,& -c^\#x^\#-d^\#\eta^\# \\
      \end{pmatrix}.
      \end{align*}
  \item $|S|=1$ and $|z|=1$.\\
      LHS:
      $$
      \begin{pmatrix}
        a & b \\
        c & d \\
      \end{pmatrix}
      \begin{pmatrix}
        \eta \\
        x \\
      \end{pmatrix}
      =
      \begin{pmatrix}
        a\eta+bx \\
        c\eta+dx \\
      \end{pmatrix}
      \toddg
      \begin{pmatrix}
        a^\#\eta^\#+b^\#x^\# ,& c^\#\eta^\#+d^\#x^\# \\
      \end{pmatrix}.
      $$
      RHS:
    \begin{align*}
      -\begin{pmatrix}
            \eta^\# ,&-x^\# \\
          \end{pmatrix}
          \begin{pmatrix}
            a^\# & -c^\# \\
            b^\# & d^\# \\
          \end{pmatrix}
       = & -\begin{pmatrix}
            \eta^\# a^\#- x^\#b^\# ,& -\eta^\# c^\#-x^\#d^\# \\
          \end{pmatrix} \\
       = & \begin{pmatrix}
            a^\#\eta^\#+b^\#x^\#  ,& c^\#\eta^\#+d^\#x^\# \\
          \end{pmatrix}.
    \end{align*}
  \end{itemize}
  (ii) The antilinearity of $\ddg$ (Eq.~(\ref{eq:gradedadjoint})) follows from the linearity of $ST$ and the properties of the hash mapping Def.~\ref{def:automorph}. \\
  (iii) Since Eq.~(\ref{eq:STcomposition}) holds for $S\in\MM$ as well~\cite{varadara,manin} and because the hash preserves the order of a product, Eq.~(\ref{eq:hashactionii}), the desired property immediatelly follows.\\
  (iv) Finally, to show Eq.~(\ref{eq:gradedadjointiii}) we use the commutativity of $ST$ and $\#$ and write $\ddg\circ\ddg=\#\circ\#\circ ST\circ ST$. If $|S|=0$ we get
  $$
  S=\begin{pmatrix}
    a & b \\
    c & d \\
  \end{pmatrix}\overset{ST}{\to}
  \begin{pmatrix}
    a & c \\
    -b & d \\
  \end{pmatrix}\overset{ST}{\to}
  \begin{pmatrix}
    a & -b \\
    -c & d \\
  \end{pmatrix}\overset{\#\circ\#}{\to}
  \begin{pmatrix}
    a & b \\
    c & d \\
  \end{pmatrix}\equiv S
  $$
  using property~(\ref{eq:hashactioniii}). Similarly, for $|S|=1$ we obtain
  $$
  S=\begin{pmatrix}
    a & b \\
    c & d \\
  \end{pmatrix}\overset{ST}{\to}
  \begin{pmatrix}
    a & -c \\
    b & d \\
  \end{pmatrix}\overset{ST}{\to}
  \begin{pmatrix}
    a & -b \\
    -c & d \\
  \end{pmatrix}\overset{\#\circ\#}{\to}
  \begin{pmatrix}
    -a & -b \\
    -c & -d \\
  \end{pmatrix}\equiv -S
  $$
  as required.
  \end{proof}
The double dagger operator used in the previous theorem will be called the {\em grade adjoint}  or eventually {\em superadjoint}. But this is not the end of the story. Whatever extension of quantum mechanics we are trying to invent, it must contain current quantum mechanics with all its successful machinery and measurement predictions. This is the main reason to focus on a certain subset of supermatrices equipped with the grade adjoint. The subset is defined by asking the supermatrices to be  super anti-Hermitian (also called super self-adjoint). They become anti-Hermitian if we restrict to even matrices thus making contact with ordinary quantum mechanics. We will later show something much stronger. First a definition~\cite{berezin,landimarmo}:
\begin{defi}\label{def:uosp12}
    The unitary orthosymplectic algebra $uosp(1|2;\C\La_N)$ is defined as
    $$
    uosp(1|2;\C\La_N)=\{S\in osp(1|2;\bbR)\otimes\C\La_N|S^\ddg=-S\}.
    $$
\end{defi}
This is a crucial definition and we will spend some time by analyzing its consequences. First of all, we have chosen the lowest-dimensional orthosymplectic algebra $osp(1|2;\bbR)$ to start with. In the previous paragraph we talked about extending quantum mechanics without sacrificing any of its properties. This is perhaps too ambitious for the first try so let's consider an extension of a basic building block of finite-dimensional quantum mechanics -- a two-level system (qubit). Arguments will be presented to show that $osp(1|2;\bbR)$ is the lowest-dimensional Lie superalgebra containing qubits with all its usual quantum-mechanical properties and the unitary orthosymplectic algebra $uosp(1|2;\C\La_N)$ is a key tool~\cite{superqubits}.

Contrary to $osp(1|2;\bbR)$, the algebra $uosp(1|2;\C\La_N)$ is {\em not} a Lie superalgebra in the sense of Def.~\ref{def:Liesuperalgebras}. It contains only even supermatrices as we prove in the next lemma. Note that sometimes $uosp(1|2;\C\La_N)$ is defined as $S\in osp(1|2;\C\La_N)$ subject to $S^\ddg=-S$~\cite{berezin,rittscheu} where the elements of $osp(1|2;\C\La_N)$ are already by definition even supermatrices. Here we show that one can start with an arbitrary `grassmannified' element of $osp(1|2;\bbR)$ given by the above definition~\cite{landimarmo} and it is the super self-adjoint constraint $S^\ddg=-S$  that singles out even supermatrices.

Another important (and perhaps surprising) consequence of $S^\ddg=-S$  is that the order of the standard basis Eq.~(\ref{eq:freebasis}) for the underlying $\bbZ_2$-graded vector space has been reversed with respect to $osp(1|2;\bbR)$. Let's illustrate it on the odd generators of the $osp(1|2;\bbR)$ algebra. Written in the standard basis Eq.~(\ref{eq:freebasis}) for $p=1$ and $q=2$, they are reresented by the following matrices~\cite{berezin,landimarmo}:
\begin{equation}\label{eq:ospOddGens}
    Q_1={1\over2}\begin{pmatrix}
       0 & -1 & 0 \\
      0 & 0 & 0 \\
      -1 & 0 & 0
     \end{pmatrix},\hspace{5mm}
    Q_2={1\over2}\begin{pmatrix}
      0 & 0 & 1 \\
      -1 & 0 & 0 \\
      0 & 0 & 0
     \end{pmatrix}.
\end{equation}
On the other hand, the $uosp(1|2;\C\La_N)$ algebra has a bosonic (even) subalgebra that happens to be the $su(2)$ algebra occupying a two-dimensional subspace spanned by (two) even basis vectors.  So in reality we should write $uosp(2|1;\C\La_N)$ instead of $uosp(1|2;\C\La_N)$. We won't do that to keep the  notation consistent with the majority of literature but to be able to consistently use the supermatrix operations and rules as presented in Appendix~\ref{sec:AppBackground} (recall that the supertranspose depends on the standard basis order), we have to work with the generators of $uosp(1|2;\C\La_N)$ written in the standard basis. The odd generators $uosp(1|2;\C\La_N)$ happen to be the same as those in Eq.~(\ref{eq:ospOddGens}) and that implies to shift the basis order and redefine the odd generators as
\begin{subequations}\label{eq:uospOddGens}
\begin{align}
  Q_1 & \mapsto UQ_1U^T=
  {1\over2}\begin{pmatrix}
    0 & 0 & 0 \\
    0 & 0 & -1 \\
    -1 & 0 & 0 \\
  \end{pmatrix},
   \\
  Q_2 & \mapsto UQ_2U^T=
  {1\over2}\begin{pmatrix}
    0 & 0 & -1 \\
    0 & 0 & 0 \\
    0 & 1 & 0 \\
  \end{pmatrix},
\end{align}
\end{subequations}
where
$$
U=\begin{pmatrix}
    0 & 1 & 0 \\
    0 & 0 & 1 \\
    1 & 0 & 0 \\
  \end{pmatrix}.
$$
We will label them $Q_1,Q_2$ as well. The same transformation applies to the even generators of $osp(1|2;\bbR)$  written~\cite{landimarmo} as $A_j=i/2(0\oplus\s_j)$  where $\s_i$ are Pauli matrices: $A_j\mapsto UA_jU^T=i/2(\s_j\oplus0)$. This is the convention used in the rest of the paper.
\begin{lem}\label{lem:evenmatrices}
The elements of  $uosp(1|2;\C\La_N)$ are even supermatrices.
\end{lem}
\begin{proof}
The Lie superalgebra $osp(1|2;\bbR)$ has five generators. Following Def.~\ref{def:endos} and the discussions preceding this lemma, we see that $|A_j|=0$ and $|Q_i|=1$ and they belong to ${\rm End}(\bbR^{2|1})\subset\M(2|1;\C\La_N)$. The algebra $uosp(1|2;\C\La_N)$ also forms a subset of $\M(2|1;\C\La_N)$. We will assume the existence of odd supermatrices in $uosp(1|2;\C\La_N)$ and prove the statement by contradiction. \\
(i) Assume $S=\sum_iz_i\zeta_iA_i$ where $z_i\in\bbC,|\zeta_i|=1$ and so $|S|=1$. Then from $S^\ddg=-S$ and $A_i^\ddg=-A_i$ it follows
$$
\sum_i\bar{z}_i\zeta_i^\# = \sum_iz_i\zeta_i.
$$
But this is impossible unless $z_i=0$ since for all odd supernumbers (and from any finite-dimensional Grassmann algebra) Eq.~(\ref{eq:hashaction}) dictates that $(\zeta_i^\#)^\#=-\zeta_i$.\\
(ii) Now suppose that $S=\sum_iz_i\zeta_iQ_i$ where  $z_i\in\bbC$ and $|\zeta_i|=0$ (so $|S|=1$ again). Let $\epsilon^{ij}$ be the two-dimensional anti-symmetric tensor ($\epsilon^{12}=1$). One can verify from (\ref{eq:ospOddGens}) that $Q_i^\ddg\equiv Q_i^{ST}=-\epsilon^{ij}Q_j$ holds. But then again, this is incompatible with the constraint $S^\ddg=-S$  unless $z_i=0$.
\end{proof}
Hence an arbitrary  element of the $uosp(1|2;\C\La_N)$ algebra is given by
\begin{equation}\label{eq:uosparbitraryelement}
    S=\xi_1A_1+\xi_2A_2+\xi_3A_3+\zeta Q_1+\zeta^\# Q_2,
\end{equation}
since $S^\ddg=-S$ is valid as long as $\xi_i^\#=\xi_i\in\C\La_{N,0}$ holds and for any $\zeta\in\C\La_{N,1}$.

As we have already mentioned~\cite{berezin_book,berezin,rittscheu}, the $uosp(1|2;\C\La_N)$  algebra is not a Lie superalgebra but something closer to an ordinary Lie algebra due to the presence of  Grassmann numbers. The main purpose for considering even supermatrices such as Eq.~(\ref{eq:uosparbitraryelement}) is that unlike the case of Lie superalgebras, there exists an exponential map transforming this Grassmann number-assisted Lie algebra into the corresponding Lie group~\cite{berezin,kostant}. More precisely, there is an equivalent of the Zassenhaus formula for even supermatrices but not for orthosymplectic Lie superalgebras due to the presence of the anticommutator for odd elements of $osp(p|q;\bbR)$~\cite{bch}. A similar issue seems to exist for other Lie superalgebras.

\subsubsection*{The unitary supergroup $UOSP(1|2;\C\La_N)$ and superqubits}

We start with the definition of the $UOSP(1|2;\C\La_N)$ group~\cite{berezin}.
\begin{defi}\label{def:UOSPgroup}
    The $UOSP(1|2;\C\La_N)$ group is defined as
    $$
    UOSP(1|2;\C\La_N)=\{Z=\exp{S}|S\in uosp(1|2;\C\La_N)\},
    $$
    where $S$ is Eq.~(\ref{eq:uosparbitraryelement}). An arbitrary group element can be written as
    \begin{equation}\label{eq:UOSPoperator}
        Z=\exp{[\xi_1A_1+\xi_2A_2+\xi_3A_3]}\exp{[\zeta Q_1+\zeta^\# Q_2]}.
    \end{equation}
\end{defi}
Note that the super adjoint condition on the algebraic level leads to the superunitary condition on the group level
\begin{equation}\label{eq:superunitary}
    Z^\ddg Z=ZZ^\ddg=1.
\end{equation}
It is no coincidence that it resembles the pattern from the ordinary $su(d)$ Lie algebra. A self-adjoint generator of the $su(d)$  Lie algebra becomes a unitary matrix representing an element of the corresponding group $SU(d)$.

We  set $N=2$ for the order of the Grassmann algebra $\C\La_N$. This is the lowest-dimensional non-trivial complex Grassmann algebra equipped with the grade involution (the hash map $\#$) from Def.~\ref{def:automorph}. The case $N=1$ is impossible since at least two odd Grassmann generators are needed ($\eta$ and its complex conjugate $\eta^\#$). For $N=0$ the whole process is a mere complexification we are not interested in. Hence, following Def.~\ref{def:supernumber}, the $S^\ddg=-S$ condition dictates the most general form of coefficients in Eq.~(\ref{eq:uosparbitraryelement}) to be $\xi_i=a_i+b_i\eta\eta^\#$  and $\zeta=p_1\eta+p_2\eta^\#$ where $p_1,p_2\in\bbC$.

The constraint $S^\ddg=-S$ implies $a_i,b_i\in\bbR$. Recall that $A_i^\ddg=-A_i$ and $|\xi_i|=0$. Hence $\xi_i^\#=\xi_i$ and so $\bar a_i=a_i$ and $\bar b_i=b_i$. Surprisingly, $p_1,p_2$ in $\zeta$ can't be arbitrary complex but the reason is not $S^\ddg=-S$. We will get to it in Lemma~\ref{lem:Ssimplified}.

To proceed, we take an inspiration from the world of qubits. Qubits carry the fundamental representation of $SU(2)$. The space of qubits is not identified with the $SU(2)$ group manifold (the $S^3$ sphere) but rather with a coset space $S^2=SU(2)/U(1)$ (the Bloch sphere). The reason is that from the physical point of view, there is a redundancy in the form of an overall phase generated by $U(1)$. This  well known insight is based on the geometric approach to quantum mechanics~\cite{qubit}, but let's make it explicit to compare it with what follows for superqubits. If we exponentiate an arbitrary element of the $su(2)$ Lie algebra we obtain
$$
    V=\exp{\sum_ia_iA_i}=\cos{m\theta\over2}\openone+{i\over m}\sin{m\theta\over2}(a_1\s_1+a_2\s_2+a_3\s_3),
$$
where $m=\sum_{i}a_i^2$. The explicit transition from $SU(2)$ to $SU(2)/U(1)$ is achieved by setting $m=1$.

For the $SU(2)$ part of $UOSP(1|2;\C\La_2)$ the exponentiation goes through in exactly the same way. The only difference is that the parameter $m=\sum_{i}a_i^2+2a_ib_i\eta\eta^\#$ is even Grassmann. We can ignore this overall Grassmann number by setting $m=1$~\cite{landimarmo,cohstates,supersphere,supersphere1} and so
\begin{equation}\label{eq:condsOnSuperqubit}
\sum_{i}a_i^2=1\mbox{\ \ and\ \ }\sum_{i}a_ib_i=0.
\end{equation}
\begin{defi}[\cite{superqubits,landimarmo,cohstates,supersphere,supersphere1,superspinors}]\label{def:superQ}
    A $\C\La_2$-superqubit is a carrier of the fundamental representation of the group~$UOSP(1|2;\C\La_2)$.
\end{defi}
$\C\La_2$-superqubits will simply be called superqubits. By performing the exponentiation in Eq.~(\ref{eq:UOSPoperator}) (see examples in Appendix~\ref{sec:AppSLinearALg} where the superlinear algebra calculations are illustrated  on it) we find
\begin{align}\label{eq:UOSPgroup}
    Z(2p_1\eta,2p_2\eta^\#,\a,\b)&=U(\a,\b)S(2p_1\eta,2p_2\eta^\#)\nn\\
     &= \begin{pmatrix}
        \a & -\b^\# & 0 \\
        \b & \a^\# & 0 \\
        0 & 0 & 1
      \end{pmatrix}
    \begin{pmatrix}
      1+{P^2\over2}\eta\eta^\# & 0 & -\bar p_1\eta^\#+\bar p_2\eta \\
      0 & 1+{P^2\over2}\eta\eta^\# & -p_1\eta-p_2\eta^\# \\
      p_1\eta+p_2\eta^\# & -\bar p_1\eta^\#+\bar p_2\eta & 1-P^2\eta\eta^\# \\
    \end{pmatrix},
\end{align}
where $P^2=|p_1|^2+|p_2|^2$. If we set $b_i=0$ for all $i$ in Eq.~(\ref{eq:condsOnSuperqubit}) we can interpret $\a=\cos{\vt},\b=e^{i\phi}\sin{\vt}$ as the usual reparametrization of the Bloch sphere for qubits since then $\a,\b\in\bbC$ and so $\a^\#\equiv\bar\a,\b^\#\equiv\bar\b$.
\begin{lem}\label{lem:Ssimplified}
  For $Z$ in Eq.~(\ref{eq:UOSPgroup}) to belong to $UOSP(1|2;\C\La_2)$, the parameters $p_1,p_2$ must satisfy $p_1=p_2\equiv p$ where $p\in\bbR$.
\end{lem}
\begin{proof}
  Since $Z\in UOSP(1|2;\C\La_2)$ then for any other $Z'(2q_1\eta,2q_2\eta^\#,\g,\d)\in UOSP(1|2;\C\La_2)$ it must hold $ZZ'=Z''\in UOSP(1|2;\C\La_2)$ as well. For the purpose of the proof let's further assume $U(\a,\b)=U(\g,\d)=\id$, where $\id$ is the unit (identity) matrix. Then the above requirement becomes $SS'=S''$ which can be rephrased using Eq.~(\ref{eq:UOSPoperator}) as
  \begin{equation}\label{eq:expProduct}
    \exp{[\zeta Q_1+\zeta^\# Q_2]}\exp{[\la Q_1+\la^\# Q_2]}=\exp{[(\zeta+\la) Q_1+(\zeta^\#+\la^\#) Q_2]},
  \end{equation}
  where $\zeta=p_1\eta+p_2\eta^\#$ and $\la=q_1\eta+q_2\eta^\#$ ($p_1,p_2,q_1,q_2\in\bbC$). Using Eq.~(\ref{eq:zetaANDoddGens}), the LHS of Eq.~(\ref{eq:expProduct}) becomes
  \begin{align}\label{eq:LHS}
    LHS = & \id+(\zeta Q_1+\zeta^\# Q_2)+(\la Q_1+\la^\# Q_2)-{1\over2}\zeta\zeta^\#Q-{1\over2}\la\la^\#Q\nn\\
     - & \zeta\la Q_1^2-\zeta^\#\la^\# Q_2^2-\zeta^\#\la Q_2Q_1-\zeta\la^\# Q_1Q_2,
  \end{align}
  where $Q=Q_1Q_2-Q_2Q_1$ and the terms with more than two Grassmann numbers are zero due to their nilpotentcy. The RHS reads
  \begin{align}\label{eq:RHS}
    RHS = & \id+(\zeta+\la) Q_1+(\zeta^\#+\la^\#) Q_2-{1\over2}\zeta\zeta^\#Q-{1\over2}\la\la^\#Q\nn\\
     - & {1\over2}(\zeta^\#\la+\la^\#\zeta) Q_2Q_1-{1\over2}(\zeta\la^\#+\la\zeta^\#) Q_1Q_2.
  \end{align}
  The first rows on the left and right are identical. For equality~(\ref{eq:expProduct}) to hold, $Q_1^2$ and $Q_2^2$ must  vanish from Eq.~(\ref{eq:LHS}) implying
  \begin{equation}\label{eq:constr1}
      p_2q_1=p_1q_2
  \end{equation}
  and $\zeta\la^\#=-\zeta^\#\la$ dictates
  \begin{equation}\label{eq:constr2}
      p_1\bar q_1+p_2\bar q_2=\bar p_1 q_1+\bar p_2 q_2.
  \end{equation}
  But this is not enough. In addition to the two  constraints on $p_i$ and $q_j$, these coefficients must also be  independent. This is not an oxymoron. For $S,S'$ to be group elements, the action of $S$ following $S'$ cannot be dependent on the coefficients chosen for $S'$. So $p_1,p_2$ cannot be a function of $q_1$ and $q_2$ in any way. Therefore, Eqs.~(\ref{eq:constr1}) and (\ref{eq:constr2}) are to be understood as constraints on $p_i,q_j$, not prescriptions to get $p_i$ from $q_j$ or vice-versa.

  Remarkably, these requirements can be satisfied by first setting $K p_1=p_2\equiv p$ and $K q_1=q_2\equiv q$ where $K\in\bbC$.  Like that the first constraint Eq.~(\ref{eq:constr1}) is satisfied and the second becomes (by setting $p_2\equiv p,q_2\equiv q$)
  $$
  p\bar q=\bar p q.
  $$
  Here again $q$ can't depend on $p$ or $\bar p$ and so the only option is to set $p=\bar p$ and $q=\bar q$. Hence $p,q\in\bbR$.
\end{proof}
Following this result, the final step we take is to reparametrize the Grassmann variables in Eq.~(\ref{eq:UOSPgroup}) by setting
\begin{align*}
  \zeta=p_1\eta+p_2\eta^\# &\equiv p(\eta+\eta^\#)  \mapsto p\eta, \\
  \zeta^\#=\bar p_1\eta^\#-\bar p_2\eta &\equiv p(\eta^\#- \eta) \mapsto p\eta^\#.
\end{align*}
Hence $P^2=p^2$ and from Eq.~(\ref{eq:UOSPgroup}) we finally obtain
\begin{equation}\label{eq:UOSPgroupS}
    S(2p\eta)=
    \begin{pmatrix}
      1+{p^2\over2}\eta\eta^\# & 0 & - p\eta^\# \\
      0 & 1+{p^2\over2}\eta\eta^\# & -p\eta \\
      p\eta & - p\eta^\# & 1-p^2\eta\eta^\# \\
    \end{pmatrix}.
\end{equation}
The derivation of $S(2p\eta)$ with the prior knowledge of Lemma~\ref{lem:Ssimplified} is presented in the last section of Appendix~\ref{sec:AppSLinearALg}.

Since $A_i,Q_j\in{\rm End}(\bbR^{2|1})$ it follows that $Z(2p\eta,\a,\b)\in\M(2|1,\C\La_2)$. Actually, due to Lemma~\ref{lem:evenmatrices} the $Z$ supermatrices are even only. The standard basis for $\M(2|1,\C\La_2)$ reads
\begin{equation}\label{eq:superqubitBases}
    \ket{0}=\begin{pmatrix}
              1 \\
              0 \\
              0 \\
            \end{pmatrix},\hspace{4mm}
    \ket{1}=\begin{pmatrix}
              0 \\
              1 \\
              0 \\
            \end{pmatrix},\hspace{4mm}
    \kbd=\begin{pmatrix}
              0 \\
              0 \\
              1 \\
            \end{pmatrix}.
\end{equation}
This is nothing else than the standard (free) basis Eq.~(\ref{eq:freebasis}) written in the physics notation. The basis states $\ket{0}$ and $\ket{1}$ are even (bosonic) states. The basis state $\kbd$ has a distinguished notation introduced in~\cite{superqubits} to stress out that the basis state is odd (fermionic). The reason why we call them bosonic and fermionic states will be clarified later.
\begin{rem}
  Here we continue using the super ket and bra notation. For its origin and differences to the previous use in Theorem~\ref{thm:ddgIsGradeAdjoint}, see the important remark on page~\pageref{rem:superkets}.
\end{rem}

Let's return to the explicit form of the superqubit Eq.~(\ref{eq:UOSPgroup}). First, we observe that the matrix $U(\a,\b)$ is an element of the $SU(2)$ subgroup as a consequence of the $su(2)$ subalgebra of $uosp(1|2,\C\La_2)$. The second matrix $S(2p\eta)$ from Eq.~(\ref{eq:UOSPgroupS}) is more interesting. Here come two simple lemmas studying its properties.
\begin{lem}\label{lem:Ssubgroup}
The matrix $S(2p\eta)$ is an element of an Abelian group isomorphic to $(\bbR,+)$ with the group operation being addition.
\end{lem}
\begin{proof}
One can easily verify the group axioms:
\begin{enumerate}[(i)\hspace{3.3522mm}]
  \item $S(2p\eta)S(2q\eta)=S(2(p+q)\eta)$ follows from matrix multiplication.
  \item The existence of an identity $S(0)=\id$ follows by inspecting Eq.~(\ref{eq:UOSPgroupS}).
  \item $S^{-1}(2p\eta)=S^\ddg(2p\eta)\equiv S(-2p\eta)$. The second equality follows from the definition of the superadjoint, the properties of the grade involution Eqs.~(\ref{eq:hashaction}) and the supertranspose given by Eq.~(\ref{eq:STforS}). Then, $S(-2p\eta)S(2p\eta)=\id$ can be verified by matrix multiplication.
\end{enumerate}
\end{proof}
The group $(\bbR,+)$ is non-compact  -- an issue whose solution we will offer later. Next, we show that the two matrices $U$ and $S$ `essentially' commute.
\begin{lem}\label{lem:USSU}
$U(\a,\b)S(2p\tilde\eta)=S(2p\eta)U(\a,\b)$ where $\tilde\eta=\a^\#\eta^\#+\b^\#\eta$.
\end{lem}
\begin{proof}
The claim can be proved directly by matrix multiplication but it is easier (and sufficient) to show that
$$
U(\a,\b)S(2p\tilde\eta)\ket{m}=S(2p\eta)U(\a,\b)\ket{m},
$$
where $m=\bu,0,1$.
\end{proof}

As follows from Eqs.~(\ref{eq:UOSPgroup}) and (\ref{eq:UOSPgroupS}), the general form of a pure superqubit is
\begin{equation}\label{eq:superqubit}
   \ket{\psi}=S(2p\eta)U(\a,\b)\ket{0}=
   \begin{pmatrix}
     \a \left(1+{p^2\over2}\eta\eta^\#\right)\\
     \b \left(1+{p^2\over2}\eta\eta^\#\right) \\
     p(\a\eta-\b\eta^\#) \\
   \end{pmatrix}=
   \left(1+{p^2\over2}\eta\eta^\#\right)(\a\ket{0}+\b\ket{1})-p(\a\eta-\b\eta^\#)\kbd.
\end{equation}
Note the minus sign in the rightmost equation accompanying the coefficient $p(\a\eta-\b\eta^\#)$. `Pulling out' the Grassmann coefficients to the left (that's what is happening in the second equality) is dictated by the rules of multiplication of even/odd vectors and even/odd Grassmann numbers (see Eq.~(\ref{eq:zetaVSbulletKet}) and the remark in the end of this section). Using the properties of the supertranspose we find the corresponding bra vector
\begin{align}\label{eq:superqubitbra}
   \bra{\psi}
   &=
   \begin{pmatrix}
     \bar\a \left(1+{p^2\over2}\eta\eta^\#\right) ,& \bar\b \left(1+{p^2\over2}\eta\eta^\#\right), & p(\bar\a\eta^\#+\bar\b\eta)
   \end{pmatrix}\nn\\
   &=
   \left(1+{p^2\over2}\eta\eta^\#\right)(\bar\a\bra{0}+\bar\b\bra{1})+p(\bar\a\eta^\#+\bar\b\eta)\bbd.
\end{align}

Def.~\ref{def:superQ} and the strict rules of supermatrix algebra lead to some differences compared to the original definition of a superqubit~\cite{superqubits}. After the following change of variables
\begin{align*}
  -\a\eta+\b\eta^\# & \mapsto\eta, \\
  -\bar\a^\#\eta-\bar\b^\#\eta & \mapsto\eta^\#
\end{align*}
(that happens to be similar to the transformation of Grassmann variables in Lemma~\ref{lem:USSU}) the states in~(\ref{eq:superqubit}) and (\ref{eq:superqubitbra}) become
\begin{subequations}
\begin{align}\label{eq:superqubitoriginal}
    \ket{\psi} &= \left(1+{p^2\over2}\eta\eta^\#\right)(\a\ket{0}+\b\ket{1})+p\eta\kbd,\\
    \bra{\psi} &= \left(1+{p^2\over2}\eta\eta^\#\right)(\bar\a\bra{0}+\bar\b\bra{1})-p\eta^\#\bbd.\label{eq:superqubitoriginalBra}
\end{align}
\end{subequations}
The main difference between here and Ref.~\cite{superqubits} is the plus sign in the even Grassmann coefficients of $\ket{0}$ and $\ket{1}$ already visible in Eq.~(\ref{eq:UOSPgroup}).

Super density matrices can now easily be constructed. As an example, assume $\a=1,\b=0$ for simplicity. Eqs.~(\ref{eq:superqubit}) and~(\ref{eq:superqubitbra}) lead to the superdensity matrix
\begin{equation}\label{eq:superdensitymatrix}
  \vr=\kbr{\psi}{\psi}
  =\begin{pmatrix}
     1+p^2\eta\eta^\# & 0 & p\eta^\# \\
     0 & 0 & 0\\
     p\eta & 0 & p^2\eta\eta^\# \\
   \end{pmatrix}.
\end{equation}
The supertrace, Eq.~(\ref{eq:sTr}), can be rewritten as $\sTr{}(S)=\sum_{i}(-)^{|i|}s_{ii}$ where $s_{ii}=\bra{i}\!S\ket{i}$ and so $\sTr{}(\vr)=1$ as expected. This corresponds to $\brk{\psi}{\psi}=1$ calculated in one of the examples in Appendix~\ref{sec:AppSLinearALg}. The supermatrix $\vr$ is even and super Hermitian: $\vr^\ddg=\vr$.
\begin{rem}
  To make a connection with Def.~\ref{def:supermodule}, note that $\psi\in\La^{2|1}$. The components of the column supermatrix in Eq.~(\ref{eq:superqubit}) are the {\em right} coordinates of $\psi$. The reason why we write them on the left is purely a matter of habit. But to be allowed to do so we had to change the sign in the second equality of~(\ref{eq:superqubit}) following Eq.~(\ref{eq:RbimoduleRight}).
\end{rem}

\subsubsection*{The standard basis and Jordan-Schwinger representation}

The generators of the $uosp(1|2,\C\La_2)$ algebra are linear operators acting on the space $\M(2|1,\C\La_N)$ spanned by the basis $\{\ket{0},\ket{1},\kbd\}$ in this (standard) order. It is advantageous to introduce the Jordan-Schwinger (also called oscillator) representation of the $uosp(1|2,\C\La_2)$ algebra. Let's define a `vacuum' state $\ket{vac}$ by $b_1\ket{vac}=b_2\ket{vac}=f\ket{vac}=0$. In the `single-particle' sector we get
\begin{equation}\label{eq:singlesector}
  \kbd=f^\ddg\ket{vac},\hspace{4mm} \ket{0}=b_1^\dg\ket{vac}\equiv b_1^\ddg\ket{vac},\hspace{4mm}\ket{1}=b_2^\dg\ket{vac}\equiv b_2^\ddg\ket{vac}.
\end{equation}
The bosonic operators $b_i,b_i^\dg$ satisfy the canonical commutation relation $[b_i,b_i^\dg]=1$ and $f,f^\ddg$ are operators satisfying the canonical anticommutation relation $\{f,f^\ddg\}=1$. We have thus justified the name bosonic for even states and fermionic for odd states. Then the $su(2)$ subalgebra generators are represented by
\begin{subequations}\label{eq:su2schwinger}
\begin{align}
   A_1&={i\over2}(b_1^\dg b_2+b_2^\dg b_1),\\
   A_2&={1\over2}(b_1^\dg b_2-b_2^\dg b_1),\\
   A_3&={i\over2}(b_1^\dg b_1-b_2^\dg b_2),
\end{align}
\end{subequations}
where  and $A_i^\ddg\equiv A_i^\dg=-A_i$. The odd generators can be expressed as
\begin{subequations}\label{eq:osp12schwinger}
\begin{align}
   Q_1&=-{1\over2}(f^\ddg b_1+fb_2^\dg),\\
   Q_2&={1\over2}(f^\ddg b_2-f b_1^\dg).
\end{align}
\end{subequations}
One can verify that  the expression $Q_i^\ddg=-\e^{ij}Q_j$ holds in the Jordan-Schwinger operator representation as well by  using  $(f^\ddg)^\ddg=-f$ and $f^\ddg b_i=b_if^\ddg$. The first property, on the other hand, follows from the fact that $\kbd$ is an odd vector and $\kbd^{\ddg\circ\ddg}\equiv\kbd^{ST\circ ST}=-\kbd$. It is a special case of Eq.~(\ref{eq:STactsOnVectors2}) (see also the identification of superkets and bras in the subsection that follows in Appendix~\ref{sec:AppBackground}).

It remains to be argued why $\{\ket{0},\ket{1}\}$ is called a qubit basis. Naively, it seems sufficient to say that the states are eigenstates of $A_3$ that corresponds to the embedded Pauli operator $\s_3$ for qubits. But one has to investigate what happens for a system of two or more bosons and whether their exchange statistics conform to the behavior of two or more qubits when they are swapped. The answer is that unlike a system of two or more fermions, we can indeed associate $N$ distinguishable bosons with $N$ qubits. Note that in our case the parameter that distinguishes the two bosons $b_i ,(b_i^\dg)$ is the index $i$.

The space spanned by the vectors from Eq.~(\ref{eq:singlesector}) is a $\bbZ_2$-graded Hilbert space. But once we start constructing even particle sectors we run into troubles. It turns out that this is an example of a grade star representation~\cite{snr2} where by taking a tensor product of two such representations we obtain a vector space that is not a Hilbert space. It is a consequence of the following lemma.
\begin{lem}\label{lem:gradetensorprod}
\begin{equation}\label{eq:gradeproduct}
  \ddg:f_1^\ddg f_2^\ddg\mapsto f_1 f_2.
\end{equation}
\begin{proof}
  Let $\ket{\psi}_i=\g_i\ket{0}_i+p_i\eta_i\kbd_i$ be two ($i=1,2$) superqubits Eq.~(\ref{eq:superqubitoriginal}), where $\g_i=1+{p_i^2\over2}\eta_i\eta_i^\#$ and we assume $\a=1,\b=0$ that is sufficient for the proof's sake. Superqubits are even supervectors and hence $\ket{\psi}_1\ket{\psi}_2=\ket{\psi}_2\ket{\psi}_1$. This follows from elementary superlinear algebra:
    \begin{align}\label{eq:psi12ISpsi21}
      \ket{\psi}_1\ket{\psi}_2 & =(\g_1\ket{0}_1+p_1\eta_1\kbd_1)(\g_2\ket{0}_2+p_2\eta_2\kbd_2)\nn\\
      & = \g_1\g_2\ket{0}_1\ket{0}_2+\g_1p_2\eta_2\ket{0}_1\kbd_2+p_1\g_2\eta_1\kbd_1\ket{0}_2-p_1p_2\eta_1\eta_2\kbd_1\kbd_2\nn\\
      & = \g_2\g_1\ket{0}_2\ket{0}_1+p_2\eta_2\kbd_2\g_1\ket{0}_1+\g_2\ket{0}_2p_1\eta_1\kbd_1+p_1\eta_1\kbd_1p_2\eta_2\kbd_2\\
      & = (\g_2\ket{0}_2+p_2\eta_2\kbd_2)(\g_1\ket{0}_1+p_1\eta_1\kbd_1)\nn\\
      & = \ket{\psi}_2\ket{\psi}_1,\nn
    \end{align}
    where the third equality comes from the fermionic character of the bullet state: $f_1^\ddg f_2^\ddg=-f_2^\ddg f_1^\ddg$.

    Since the superqubits are normalized to one (see one of the examples from the last section of Appendix~\ref{sec:AppSLinearALg}) we can proceed by writing:
    \begin{equation}\label{eq:twoinnerproducts}
      1=\brk{\psi}{\psi}_1\brk{\psi}{\psi}_2=\bra{\psi}_1(\brk{\psi}{\psi}_2) \ket{\psi}_1=\bra{\psi}_1\bra{\psi}_2\ket{\psi}_2\ket{\psi}_1=\bra{\psi}_1\bra{\psi}_2\ket{\psi}_1\ket{\psi}_2.
    \end{equation}
    Therefore, the grade adjoint must satisfy $\ddg:\ket{\psi}_1\ket{\psi}_2\mapsto \bra{\psi}_1\bra{\psi}_2$. To see the consequence, we write (following Eqs.~(\ref{eq:superqubitoriginalBra}) and~(\ref{eq:zetaVSbulletBra}))
    \begin{align}\label{eq:braproduct}
    \bra{\psi}_1\bra{\psi}_2&=\g_1\g_2\bra{0}_{1}\bra{0}_{2}
    +\g_1p_2\eta_2^\#\bra{0}_{1}\bbd_2+g_2p_1\eta_1^\#\bbd_1\bra{0}_{2}-p_1p_2\eta_1^\#\eta_2^\#\bbd_{1}\bbd_2.
    \end{align}
    Therefore, the action of the superadjoint on the double-bullet component of $\ket{\psi}_1\ket{\psi}_2$ in Eq.~(\ref{eq:psi12ISpsi21}) must result in:
    $$
    -p_1p_2\eta_1\eta_2\kbd_1\kbd_2\equiv -p_1p_2\eta_1\eta_2f_1^\ddg f_2^\ddg\ket{vac}
    \overset{\ddg}{\to}
    -p_1p_2\eta_1^\#\eta_2^\#\bra{vac}f_1 f_2
    \equiv
    -p_1p_2\eta_1^\#\eta_2^\#\bbd_1\bbd_2.
    $$
    (The minus sign in the leftmost expression comes from swapping $\eta_2$ and $\kbd_1$ in the last summand of the middle row of~(\ref{eq:psi12ISpsi21})).  But this is precisely the last summand of Eq.~(\ref{eq:braproduct}) and hence $(f_1^\ddg f_2^\ddg)^\ddg=f_1f_2$.
\end{proof}
\end{lem}
This is a rather important lemma so let's show if it is consistent with the properties of odd operators. From Eq.~(\ref{eq:gradedadjointiii}) we know that $(f^\ddg)^\ddg=-f$ since $f$ is odd. The product of three fermion operators is odd as well and so
$$
\left((f_1f_2f_3)^\ddg\right)^\ddg=-f_1f_2f_3.
$$
Clearly, the proved action of the grade adjoint in Eq.~(\ref{eq:gradeproduct}) is compatible with the above equation because of
$$
\left((f_1f_2f_3)^\ddg\right)^\ddg=\left(f_1^\ddg (f_2f_3)^\ddg\right)^\ddg
=\left(f_1^\ddg f_2^\ddg f_3^\ddg\right)^\ddg=-f_1f_2f_3.
$$
\begin{rem}
If we used the properties of the usual adjoint that reverses the order of two operators upon which it acts, we would find that
$$
\left((f_1f_2f_3)^\dg\right)^\dg=f_1f_2f_3
$$
as expected. The two mappings (dagger and double dagger) are indeed different in many aspects.
\end{rem}

\subsubsection*{Super Hilbert space}
\label{subsubsec:superH}

Lemma~\ref{lem:gradetensorprod} has interesting consequences. The grade adjoint of $\kbd_1\kbd_2=f_1^\ddg f_2^\ddg\ket{vac}$ is $\bbd_1\bbd_2=\bra{vac}f_1f_2$ and so the norm of this state is negative
\begin{equation}\label{eq:gradestarrepproduct}
  \bbd_1\bbd_2\kbd_1\kbd_2=-\brk{\bu}{\bu}_1\brk{\bu}{\bu}_2=-1.
\end{equation}
Does it mean that after so much work we don't even have a proper Hilbert space? Fortunately, the answer is no and there are two reasons for it. First, looking at Eq.~(\ref{eq:UOSPgroup}) we notice something unusual. The $UOSP(1|2;\C\La_2)$ does not act transitively and so the superqubit space is not a homogeneous space. There is no unitary $Z\in UOSP(1|2;\C\La_2)$ that would take us from a subspace spanned by $\{\ket{0},\ket{1}\}$ to the subspace spanned by $\kbd$. That is not surprising because $Z$ is even and by definition it cannot change the degree of a homogeneous vector.

In principle, we could define even superqubits like in Eq.~(\ref{eq:superqubit}) and odd superqubits by $S(2p\eta)U(\a,\b)\kbd$ that would not be equivalent. However, a tensor product of two odd superqubits would suffer from the same problem as the state $\kbd_1\kbd_2$ -- its norm would be negative.

The second key aspect is the transition from Lie superalgebras to Grassmann-valued Lie algebras we underwent in Def.~\ref{def:uosp12}. The constraint on even operators is nothing else than a super version of antihermiticity.
We can trivially rewrite the constraint $S^\ddg=-S$ as $S^\ddg G+GS=0$ where
\begin{equation}\label{eq:G}
G=\begin{pmatrix}
    1 & 0 & 0 \\
    0 & 1 & 0 \\
    0 & 0 & 1 \\
  \end{pmatrix}
\end{equation}
is a matrix representing a non-degenerate, bilinear and positive semidefinite form that again appears right after Def.~\ref{def:UOSPgroup} in Eq.~(\ref{eq:superunitary}) as $Z^\ddg GZ=G$. We implicitly used this metric when we normalized the superqubit in Eq.~(\ref{eq:superqubit}). This choice is important  yet from another reason than positivity. The $UOSP(1|2;\C\La_2)$ group acts as an isometry group on a vector space equipped with the inner product induced by $G$.
\begin{defi}\label{def:hilbertspace}
Let $\SS$ be a non-degenerate form
$$
    \SS:V\times V\mapsto\C\La_{2,0},
$$
where $V=\La^{2|1}$ is a $\C\La_2$-bimodule (Def.~\ref{def:supermodule}). We wish the following properties to be satisfied:
\begin{enumerate}[(i)\hspace{3.3522mm}]
  \item $\SS(v,v)\geq0$ (positive semidefinite),
  \item $\SS(u,v)=\SS(v,u)^\#$ (super Hermitian),
  \item $\SS(\a u,v)=\bar\a\SS(u,v)$ (sesquilinear),
  \item $\SS(u,v)_\bbC=\SS(u_\bbC,v_\bbC)$ (consistent),
\end{enumerate}
where $u,v\in \La^{2|1}$ are even  and $\a\in\bbC$. The subscript $\bbC$ in case (iv) denotes the complex (non-Grassmann) part of a supernumber or any other encountered superstructure.
\end{defi}
The definition has interesting consequences that we will discuss in detail. 
\begin{prop}\label{prop:formproperties}
Let $\brk{_-}{_-}$  be an inner product induced by $G(e_i,e_j)=g_{ij}$ (Eq.~(\ref{eq:G})) defined as
\begin{equation}\label{eq:innerproduct}
  \brk{u}{v}=u^{i\#}v^jg_{ij}\equiv u_j^{\#}v^j,
\end{equation}
where $u,v\in\La^{2|1}$ are even and written in the standard basis: $u=e_iu^i$ and $v=e_jv^j$ where $i,j=1,2,3$ correspond to the standard basis~Eq.~(\ref{eq:superqubitBases}). The components $u^i,v^j$ are the right coordinates of $u$ and $v$ forming column supervectors. Then the inner product satisfies the properties listed in Def.~\ref{def:hilbertspace}.
\end{prop}
\begin{rem}
    It follows from definition that $\brk{Zu}{Zv}=\brk{u}{Z^\ddg Zv}=\brk{u}{v}$ for all $u,v\in \La^{2|1}$ and $Z\in UOSP(1|2;\C\La_2)$. In ordinary complex vector spaces there wouldn't be a reason to prove anything. $G$ would be a metric preserved by $SU(3)$ and the axioms from Def.~\ref{def:hilbertspace} would become trivial or satisfied by definition. What makes things less trivial is that $\La^{2|1}$ is a $\C\La_2$-bimodule (for more information see Appendix~\ref{sec:AppBackground}).
\end{rem}
\begin{rem}
  Note that on the left hand side of (\ref{eq:innerproduct}) there are elements of $\La^{2|1}$ but on the right hand side, the right coordinates of $u,v$  appear that form row and column supermatrices belonging to $\M(2|1,\C\La_2)$. The advantage of calculating on the right side is that it is mere  multiplication of rows and columns. But by performing the calculation in $\La^{2|1}$ we obviously have to get the same result:
    \begin{equation}\label{eq:innerproductLHS}
      \brk{u}{v}=(-)^{|i|}\tilde e_iu^i{^\#}e_jv^j=(-)^{|i|(|j|\oplus1)}\tilde e_ie_ju^i{^\#}v^j=(-)^{|j|(|j|\oplus1)}u_j^\#v^j=u_j^\#v^j,
    \end{equation}
  where $\tilde e_ie_j=\d_{ij}$ (cf.~Eq.~(\ref{eq:superqubitBases})). In the second equality we used the fact that $u,v$ are even. Thus $|u^i|=|i|$ and so $u^i{^\#}e_j=(-)^{|u^i||j|}e_ju^i{^\#}=(-)^{|i||j|}e_ju^i{^\#}$ holds. Also note that $|u^i{^\#}|\equiv|u^i|$.
\end{rem}
\begin{rem}
  We are allowed to use the same bracket notation as in Theorem~\ref{thm:ddgIsGradeAdjoint} since there we use the same form for the space $\La^{p|q}$ (used already in Def.~\ref{def:gradeadjoint} for its subspace $\rm{End}(\bbC^{p|q})$). If $p=2,q=1$, the $UOSP(1|2;\C\La_2)$ group that is by definition acting on a subspace of $\M(2|1,\C\La_2)$ inherits the metric and the above four properties will be shown to be satisfied. They may not hold (and in fact they don't since, for example, positive semidefiniteness is not defined) in the whole $\M(2|1,\C\La_2)$.
\end{rem}
\begin{proof}
\begin{enumerate}[(i)\hspace{3.3522mm}]
  \item Write an even element of $\La^{2|1}$ as $v=e_1v^1+e_2v^2+e_3v^3$ where $v^1,v^2\in\C\La_{2,0}$ and $v^3\in\C\La_{2,1}$. We find that the normalization requirement $\brk{v}{v}=1$ implies
        $$
        v={1\over\big(v_j^\#v^j\big)^{1/2}}e_iv^i,
        $$
        where the normalization follows from Eq.~(\ref{eq:innerproduct}) (or~(\ref{eq:innerproductLHS})). Grassmann numbers can be inverted only if the non-Grassmann part is nonzero. That implies the necessary assumptions $v^1_\bbC\neq0$ or $v^2_\bbC\neq0$. Then the normalization can be written
        $$
        {1\over\big(v_j^\#v^j\big)^{1/2}}={1\over\big(c_1+c_2\eth\eta\big)^{1/2}}={1\over\sqrt{c_1}}-{1\over2}{c_2\over\sqrt{c_1}^{3}}\eth\eta
        ={1\over\sqrt{c_1}}+{1\over2}{c_2\over\sqrt{c_1}^{3}}\eta\eth,
        $$
        where $c_1,c_2\in\bbR$. But then $v$ is precisely the superqubit from Eq.~(\ref{eq:superqubitoriginal}). Interestingly, the requirement of positive semidefiniteness singles out superqubits.
  \item For $u,v$ not necessarily normalized to one, we write
        $$
        \big(\brk{v}{u}\big)^\#=\big(v_i^{\#}u^i\big)^\#=(-)^{|i|}v_iu^i{^\#}=u^i{^\#}v_i=\brk{u}{v}.
        $$
        The second equality follows the properties of the grade involution Def.~\ref{def:automorph} and the third equality is valid for both $v^i$ even and odd. If $v_i$ is even then $u^i{^\#}$ is even as well (recall that $g_{ij}$ is diagonal), $(-)^{|i|}=1$ and they commute. If $v_i$ is odd then $u^i{^\#}$ is odd as well, $(-)^{|i|}=-1$ and they anticommute. This cancels the minus sign and so we always obtain $u^i{^\#}v_i$.
  \item This immediately follows from the definition of the inner product and the fact that $\#$ acts as ordinary complex conjugation for
        $\a\in\bbC$.
  \item $v^1,v^2$ are even Grassmann numbers and so they also contain purely complex components by  assuming $v^1_\bbC\neq0$ or $v^2_\bbC\neq0$.
        The same holds for their product and hence
        $\brk{u}{v}_\bbC=(u_1^{\#}v^1+ u_2^{\#}v^2)_\bbC=\bar u_{1,\bbC}v^1_{\bbC}+\bar u_{2,\bbC}v^2_{\bbC}$. On the other hand, we immediately get $\brk{u_\bbC}{v_\bbC}=\bar u_{1,\bbC}v^1_{\bbC}+\bar u_{2,\bbC}v^2_{\bbC}$.
\end{enumerate}
\end{proof}
The consistency condition (case (iv)) has an impact from the physical point of view:
\begin{cor}
    Quantum theory based on superqubits is not a modification of quantum mechanics but rather its extension to a specific supersymmetric domain.
\end{cor}
Indeed, if we threw away the Grassmann part of the superqubit state (by setting $p=0$ in Eq.~(\ref{eq:superqubit})) the state reduces to an ordinary qubit requiring no further action) and all the super structures we introduced would become the familiar constructions from quantum theory. The kind of supersymmetry we study here simply describes ordinary quantum theory but in the supersymmetric domain. So it does not alter its non-supersymmetric part -- there is no reason to modify the non-supersymmetric quantum mechanics since its validity has been verified. To complete the `proof' of the corollary one last thing remains to be clarified -- the definition of the measurement probability.

Before we delve into the discussions of how to interpret Grassmann variables as probabilities we have to close the question of the existence of a Hilbert space for multipartite states. Def.~\ref{def:hilbertspace} only explicitly talks about single superqubits and the question can only be fully resolved after higher Lie superalgebras and groups have been studied. But we can say something already now. If we take a tensor product of $k$ superqubits it is clear that they will live in a subspace superunitarily connected with the usual $k$-qubit basis of dimension $2^k$. So, for instance, the state $\kbd_1\kbd_2\equiv\ket{\bu\bu}_{12}$ that causes so much trouble is not a valid two-superqubit state -- there is no superunitary $Z_1\otimes Z_2$ that would transform any state from the even two-superqubit subspace to $\ket{\bu\bu}_{12}$. Another example would be the state  $\ket{\psi}_{12}=\ket{01}-\ket{10}+\ket{\bu\bu}$. Its norm equals one but this state does not belong to the Hilbert space introduced above because the norm of one of the basis states ($\ket{\bu\bu}_{12}$) is minus one.

\subsection*{Grassmann-valued probabilities and the Rogers norm}

\begin{defi}\label{def:grassmannVALUEDprob}
    Let the Grassmann-valued transition probability function between two superqubits $\vp$ and $\psi$ be defined as
    \begin{equation}\label{eq:prob}
    p_\G(\vp,\psi)=\brk{\vp}{\psi}\big(\brk{\vp}{\psi}\big)^\#.
    \end{equation}
\end{defi}
The rationale behind the definition is easy to uncover. For ordinary qubits, Eq.~(\ref{eq:prob}) automatically becomes Born's rule. This is essential because recall that the $SU(2)$ group is a subgroup of the $UOSP(1|2;\C\La_N)$. Using item (ii) of Proposition~\ref{prop:formproperties}, we can write
$p_\G(\vp,\psi)=\brk{\vp}{\psi}\!\brk{\psi}{\vp}$ but this reminds us of the supertrace operation illustrated on Eq.~(\ref{eq:superdensitymatrix}). In fact, this expression behaves as we are used to from quantum mechanics: $\brk{\vp}{\psi}\!\brk{\psi}{\vp}=\brk{Zi}{\psi}\!\brk{\psi}{Z^\ddg i}=\brk{i}{Z^\ddg \psi}\!\brk{Z\psi}{i}=\brk{i}{\tilde\psi}\!\brk{\tilde\psi}{i}$, where $Z\in UOSP(1|2;\C\La_2)$, so  we obtained the diagonal coefficients of $\kbr{\tilde\psi}{\tilde\psi}$. But because superqubits are by definition even, this is just a special case of the supertrace rule
\begin{equation}\label{eq:OrthoBasisMeasurement}
  p_\G(i,\psi)=(-)^{|i|}\brk{i}{\tilde\psi}\!\brk{\tilde\psi}{i},
\end{equation}
where $i=\{0,1,\bu\}$ and $\sum_{i=0,1,\bu}p_\G(i,\psi)=\sTr{}(\kbr{\tilde\psi}{\tilde\psi})=1$.

The question we face now is how to interpret Grassmann-valued transition probability functions. Is there something special about the Grassmann numbers that are obtained by means of Eq.~(\ref{eq:prob}) for any two superqubits $\vp,\psi$? It turns out that this is the case. Such Grassmann numbers are not only always even (follows from Proposition~\ref{prop:formproperties} (ii)) but also satisfy the `reality condition':
\begin{defi}\label{def:realitycondition}
    An even Grassmann number $\zeta\in\C\La_{N,0}$ will be called real if $\zeta^\#=\zeta$.
\end{defi}
As a small detour, due to the reality condition definition we can actually gain some fresh insight into the origin of the two types of automorphisms from Def.~\ref{def:automorph}.
\begin{lem}\label{lem:conjsfromreality}
    Let $\mathrm{conj}:\La_N\mapsto\La_N$ be an antilinear map. For an arbitrary $\zeta\in\La_{N}$ we define the reality condition on $\zeta$ to be
    \begin{equation}\label{eq:realitycond}
        \zeta\conj{\zeta}=\conj{\zeta\conj{\zeta}}.
    \end{equation}
    Then, there are at least two types of conjugations satisfying the reality condition.
\end{lem}
\begin{proof}
First suppose that the map is an antiautomorphism. The left side of Eq.~(\ref{eq:realitycond}) becomes
\begin{equation}\label{eq:antiauto}
        \conj{\zeta\conj{\zeta}}=\conj{\conj{\zeta}}\conj{\zeta}.
\end{equation}
For it to be equal to the RHS of Eq.~(\ref{eq:realitycond}), $\mathrm{conj}$ must be an involution:
\begin{equation}\label{eq:involution}
    \conj{\conj{\zeta}}=\zeta.
\end{equation}
The map is then the usual complex conjugation defined for Grassmann variables in quantum field theory of fermions~\cite{qft} (the star map from Def.~\ref{def:automorph})
$$
(\zeta^*)^*=\zeta.
$$

The second option is an order preserving type of conjugation
\begin{equation}\label{eq:orderpreserving}
        \conj{\zeta\conj{\zeta}}={\conj{\zeta}}\conj{\conj{\zeta}}.
\end{equation}
In order to satisfy the RHS of Eq.~(\ref{eq:realitycond}) it  must hold that
\begin{equation}\label{eq:hashinvolution}
    \conj{\conj{\zeta}}=-\zeta.
\end{equation}
So this kind of conjugation is precisely the hash map also introduced in Def.~\ref{def:automorph} and used throughout this work
$$
(\zeta^\#)^\#=-\zeta.
$$
For $\zeta\in\bbC$ both maps become ordinary complex conjugation and Eq.~(\ref{eq:realitycond}) is trivially satisfied.
\end{proof}
\begin{rem}
It might be interesting to show how many more mappings there are for Grassmann variables that satisfy the reality condition.
\end{rem}
Let's go back to the interpretation of Grassmann variables. We are not the first ones to ask about their meaning~\cite{nieto}. The pioneering work in this direction had been done by A. Rogers and others in the 80's~\cite{rogersnorm}. The motivation there was then the burgeoning field of superanalysis on supermanifolds~\cite{supermani,rabincrane,cookfulp} as a response to the discovery of supersymmetric theories in physics. This is a branch of mathematics on its own indirectly related to the topic of this work. We will just define the Rogers prescription of how to extract ordinary numbers from Grassmann numbers and see if it can be of use for us. Of course, the reason why Grassmann numbers cannot be used directly is that they cannot be ordered in the first place. But there is another, closely related,  reason. The outputs of measurement devices are real numbers as well as the outcomes probabilities and we would like to have an elegant prescription \`{a} la quantum mechanics.
\begin{defi}[The Rogers norm~\cite{rogersnorm}]\label{def:rogers}
    Let $\zeta\in\C\La_N$ be an arbitrary supernumber whose general form was introduced in Def.~\ref{def:supernumber}. The Rogers norm of $\zeta$ is defined as
    \begin{equation}\label{eq:rogersnorm}
        |\zeta|_{R_1}\overset{\rm df}{=}|z_0|+\sum_{k=1}^N\sum_{m=1}^{N\choose k}|z^{(m)}|.
    \end{equation}
\end{defi}
Spaces equipped with the Rogers norm  teem with many interesting properties we will not discuss here~\cite{rogersnorm}. From a broader point of view, it is probably the most straightforward way of extracting real numbers from Grassmann numbers -- one simply looks at the accompanying coefficients. So even though our situation is different (we want to interpret even Grassmann-valued probabilities), the most natural way we will use to get real numbers from Grassmann numbers is similar.

If we applied the Rogers norm directly to the Grassmann-valued transition probability calculated according to Def.~\ref{def:grassmannVALUEDprob}
\begin{equation}\label{eq:GrassProb}
p_\G(\vp(q),\psi(p))=1+(p-q)^2\eta\eta^\#
\end{equation}
obtained from
\begin{align}\label{eq:GrassTransProb}
  \brk{\vp(q)}{\psi(p)}&=
     \begin{pmatrix}
     \bar\g \left(1+{q^2\over2}\eta\eta^\#\right) ,& \bar\d \left(1+{q^2\over2}\eta\eta^\#\right), & q(\bar\g\eta^\#+\bar\d\eta)
     \end{pmatrix}
     \begin{pmatrix}
     \a \left(1+{p^2\over2}\eta\eta^\#\right)\\
     \b \left(1+{p^2\over2}\eta\eta^\#\right) \\
     p(\a\eta-\b\eta^\#) \\
   \end{pmatrix} \nn\\
   &=1+{1\over2}(p-q)^2\eta\eta^\#
\end{align}
(and assumed $\a=\g=1,\b=\d=0$ for a moment), we would get
\begin{equation}\label{eq:pGRogers}
  |p_\G(\vp(q),\psi(p))|_{R_1}=1+(p-q)^2
\end{equation}
but face a problem: The Rogers norm does not respect the order of Grassmann variables and as a consequence we would get real number impossible to interpret as probabilities. Notice that we get the same result if we swap the Grassmann generators $|1-(p-q)^2\eta^\#\eta|_{R_1}=1+(p-q)^2$. This follows from how the Rogers norm has been defined. So a slight modification of the Rogers norm has been proposed in~\cite{bbd} where the two main differences are: (i) the modified Rogers norm respects the order of Grassmann generators that must be fixed during the whole calculation and (ii) the modified Rogers norm transforms even Grassmann-valued probability functions. This enables us to recast the calculation of the modified Rogers norm into a form familiar from the path integral formulation of QFT -- a Berezin (also called Grassmann) integral~\cite{qft}. This is the approach taken in this work -- we reformulate the modified Rogers norm as a Berezin integral.

Let's recall some of its basic properties. We will assume existence of finite-dimensional Grassmann algebras where $N=2k$ for $1\leq k<\infty$. Literature on the fermion path integral is divided regarding the definition of Grassmann integral~\cite{qft,qft1,qft2,qft3}. This is due to how a complex Grassmann algebra can be understood. Let's elaborate on this issue a bit more first by using the star involution from Def.~\ref{def:automorph}. Usually, one starts with a real Grassmann algebra of order $2k$ generated by $\{\theta_i\}_{i=1}^{2k}$ and define the single-variable Grassmann integral
$$
\int{\rm d}\theta_i\theta_j\overset{\rm df}{=}\d_{ij}.
$$
The algebra can be complexified
\begin{align*}
  \eta_j &={1\over\sqrt{2}}(\theta_j+i\theta_{j+k})  \\
  \eta^*_j &={1\over\sqrt{2}}(\theta_j-i\theta_{j+k}),
\end{align*}
where $j=1\dots k$. Hence~\cite{qft1,qft3}
\begin{align*}
  {\rm d}\eta_j &= {1\over\sqrt{2}}({\rm d}\theta_j-i{\rm d}\theta_{j+k}) \\
  {\rm d}\eta^*_j &= {1\over\sqrt{2}}({\rm d}\theta_j+i{\rm d}\theta_{j+k}),
\end{align*}
such that
$$
\int{\rm d}\eta_j\eta_j=\int{\rm d}\eta^*_j\eta^*_j=1
$$
is satisfied. Therefore
\begin{equation}\label{eq:starGrass}
  \int{\rm d}\eta_j{\rm d}\eta^*_j(-\eta_j\eta^*_j)=1
\end{equation}
and more generally for the multivariate case~\cite{qft,qft1}
\begin{equation}\label{eq:starGrassExp}
  \int\prod_{j=1}^k{\rm d}\eta_j{\rm d}\eta^*_j\exp{\bigl(-\eta_jA_{ji}\eta^*_i\bigr)}=\det{A}.
\end{equation}
This expression has been used to extract real numbers from Grassmann-valued functions in~\cite{castell}. We want the same prescription but for the grade involution $\#$ (for other options see~\cite{rudolph}).

We again define
\begin{equation}\label{eq:starGrassdef1}
\int{\rm d}\eta_i\eta_j=\d_{ij}
\end{equation}
but this implies
\begin{equation}\label{eq:starGrassdef2}
\int{\rm d}\eta^\#_i\eta^\#_j=\d_{ij}
\end{equation}
using the hash property Eq.~(\ref{eq:hashactionii}) from Def.~\ref{def:automorph}. It follows that
\begin{equation}\label{eq:superstarGrass1}
  1=\int {\rm d}\eta_j\eta_j\int {\rm d}\eta^\#_j\eta^\#_j=\int{\rm d}\eta_j{\rm d}\eta^\#_j(-\eta_j\eta^\#_j)
\end{equation}
exactly as for the star map Eq.~(\ref{eq:starGrass}).

\begin{defi}[The modified Rogers norm~\cite{bbd,castell}]\label{def:modrogers}
    Let $\tau\in\C\La_{N,0}$ be an even Grassmann number. The modified Rogers norm of $\tau$ is defined as
    \begin{equation}\label{eq:modifiedRogers}
      \big|\tau\big|_{R}\overset{\rm df}{=}\int{\rm d}^{2N}\eta\prod^{N/2}_{i=1}e^{-\eta_i\eta^\#_{i}}\tau ,
    \end{equation}
    where ${\rm d}^{2N}\eta\overset{\rm df}{=}\prod^{N/2}_{i=1}{\rm d}\eta_i{\rm d}\eta_i^\#$ and
    $\int{\rm d}^{2N}\eta\prod_{i}\exp{(-\eta_i\eta^\#_{i})} =1$.
\end{defi}
Recall that we consider Grassmann algebras where $N=2k$ for $1\leq k<\infty$.
\begin{exa}
Let's take the lowest dimensional case $N=2$ and calculate the modified Rogers norm of Eq.~(\ref{eq:GrassProb}) $p_\G(\vp(q),\psi(p))=1+(p-q)^2\eta\eta^\#=\tau$:
\begin{align}\label{eq:examplemodRog}
    p(\vp(q),\psi(p))=\big|\tau\big|_{R}&=\int{\rm d}\eta{\rm d}\eta^\#(1-\eta\eta^\#)(1+(p-q)^2\eta\eta^\#)\nn\\
    &=\int{\rm d}\eta{\rm d}\eta^\#(-\eta\eta^\#)+\int{\rm d}\eta{\rm d}\eta^\#\eta\eta^\#(p-q)^2\nn\\
    &=1-(p-q)^2.
\end{align}
\end{exa}
The transition probability between two completely general pure supequbits reads
\begin{equation}\label{eq:transprobgeneral}
  p(\vp(q,\g,\d),\psi(p,\a,\b))=(\a\bar\g+\b\bar\d)(\bar\a\g+\bar\b\d)(1-(p-q)^2)
\end{equation}
coming from Eq.~(\ref{eq:GrassTransProb}). The product $(\a\bar\g+\b\bar\d)(\bar\a\g+\bar\b\d)$ has its origin in the $SU(2)$ subgroup of Eq.~(\ref{eq:UOSPgroup}). The rest of Eq.~(\ref{eq:transprobgeneral}) is the consequence of the other subgroup isomorphic to the group $(\bbR,+)$ (the reals with addition) represented by the $S$ matrix.

We got rid of Grassmann variables by the prescription given in Def.~\ref{def:modrogers} but another problem has appeared.  We found in Lemma~\ref{lem:Ssubgroup} that the Abelian group whose elements are $S(2p\eta)$ is non-compact by looking at Eq.~(\ref{eq:transprobgeneral}) we see why it is indeed a problem. There exists a choice of $p$ ad $q$ such that the transition probability becomes negative. The probability function is meaningful only for $0\leq p(\vp,\psi)\leq1$ implying $|p-q|\leq1$. This region is depicted on the left side of Fig.~\ref{fig:fidplot}.
\begin{figure}[t]
    \begin{center}
    \resizebox{12cm}{6cm}{\includegraphics{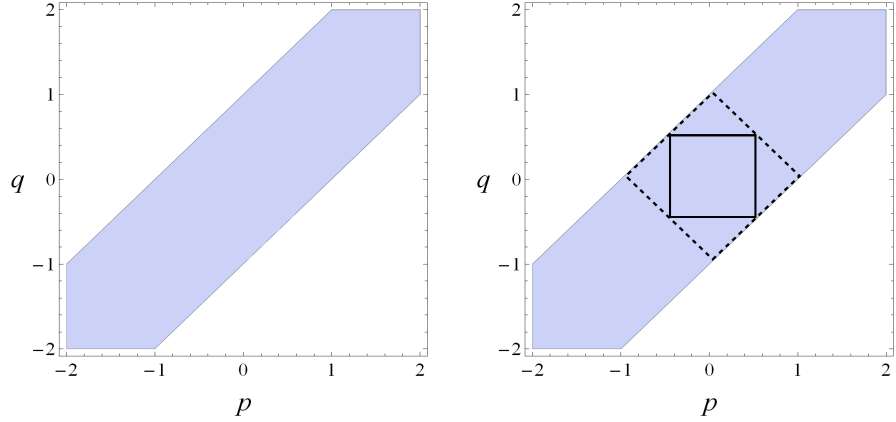}}
    \caption{The infinite blue stripe is the set where the transition probability between these two superqubits lies between zero and one. The two rectangles on the right indicate two subsets of $\bbR\times\bbR$ as candidates of how to compactify the superqubit space. The dashed rectangle is the set $s_1$ defined in Eq.~(\ref{eq:POSsubset1}). The inner rectangle is the subset $s_2$ (Eq.~(\ref{eq:POSsubset2})) motivated by Def.~\ref{def:physicalstates} dealing with the properties of transition probability functions.}
    \label{fig:fidplot}
    \end{center}
\end{figure}

Note that the probability of measurement of a superqubit $\psi(p,\a,\b)$ in the canonical basis Eq.~(\ref{eq:singlesector}) is reasonable for $0\leq|p|\leq1$. This motivates the following subset of allowed states
\begin{equation}\label{eq:POSsubset1}
s_1=\big\{p,q\in\bbR;|p-q|\leq1\cap|p+q|\leq1\big\}.
\end{equation}
The $s_1$ is the rectangle demarcated by the dashed line on the right side of Fig.~\ref{fig:fidplot}. This choice is not satisfactory though. If we set $1/2\leq |p|\leq1$ for the measured state, then there exists a rotation of the canonical basis (in particular by $S(2q\eta)$ with $1/2\leq |q|\leq1$) such that the probability is  negative again. In other words, for a given state
it makes sense to talk about measurement in one basis but not in a rotated one. This is conceptually hard to accept and to avoid this problem we further restrict the set $s_1$ to
\begin{equation}\label{eq:POSsubset2}
    s_2=\big\{p,q\in\bbR;|p|\leq1/2\cap|q|\leq1/2\big\}.
\end{equation}
The set $s_2$ is motivated by the following definition.
\begin{defi}[Physical states]\label{def:physicalstates}
    Let two superqubits $\psi(p)$ and $\vp(q)$ satisfy $0\leq|p_\G(\vp(q),\psi(p))|_{R}\leq1$. The states are considered physical only for such $p,q$ also satisfying
    \begin{equation}\label{eq:physicalLa2}
        0\leq|p_\G(\vp(\pm q),\psi(\mp p))|_{R}\leq1.
    \end{equation}
\end{defi}
The definition ensures that for all $p$ there will be $q$ from the same interval such that the transition probability between the corresponding states lies between zero and one. This naturally introduces a Cartesian product $P=P_{\psi}\times P_{\vp}$ of two positivity domains $P_{\psi}$ and  $P_{\vp}$. The positivity domain $P_\psi\subset D_\psi$ where the set $D_\psi$ is defined as
$$
D_{\psi}=\{p\in\bbR;0\leq|p_\G(\ket{\psi(p)},\ket{i})|_{R}\leq1\}.
$$
Def.~\ref{def:physicalstates} leads to
$$
P_\psi=\{p\in\bbR;-1/2\leq p\leq1/2\}
$$
and similarly for $P_\vp$. This is consistent with the set $s_2$ in Eq.~(\ref{eq:POSsubset2}) and the set is outlined by the solid rectangle in Fig.~\ref{fig:fidplot} (on the right).

We have cut a closed and bounded subset from $\bbR$ where our super evolution is allowed to take place and this amounts to compactifying the original superqubit space -- the sets $P_\psi,P_\vp$ are compact manifolds with boundary. Another virtue of Def.~\ref{def:physicalstates} is that for every $S(2p\eta)$ there exists $S^\ddg(2p\eta)$. This is because $S^\ddg(2p\eta)=S^{-1}(2p\eta)=S(-2p\eta)$ as we have noticed in Lemma~\ref{lem:Ssubgroup}. But not all group axioms are satisfied after we restricted the superqubit evolution to $P_\psi$. We know from Lemma~\ref{lem:Ssubgroup} that $S(2p_1\eta)S(2p_2\eta)=S(2(p_1+p_2)\eta)$ but what if $|p_1+p_2|>1/2$? The group law of addition is not defined beyond the domain $P_\psi$. Here we propose a solution based on the fact that $(\bbR,+)$ is a universal cover of the compact group $U(1)$. The explicit onto map is the modulo $2\pi$ function ${}\bmod{2\pi}:(\bbR,+)\mapsto U(1)$ that can be written as
\begin{equation}\label{eq:RtoU1}
  p\bmod{2\pi}=p-2\pi\Bigl\lfloor{p\over2\pi}\Bigr\rfloor
\end{equation}
valid for all $p\in\bbR$. If we make the following substitution
\begin{equation}\label{eq:subst}
  p\mapsto {p\over2\pi}-\Bigl\lfloor{p\over2\pi}\Bigr\rfloor-{1\over2}
\end{equation}
in Eq.~(\ref{eq:superqubit}) we obtain a superqubit with the right properties.
\begin{rem}
Perhaps there is a question why we bothered with Def.~\ref{def:physicalstates} if now we again compactified the whole $\bbR$. Def.~\ref{def:physicalstates} helped us to find where exactly we have to impose the periodic boundary conditions. If we imposed the periodic boundary conditions on the positivity interval leading to $s_1$ we would encounter various inconsistencies~\footnote{An explicit example exists due to Markus M\"uller.}.
\end{rem}
\begin{rem}
The mapping Eq.~(\ref{eq:RtoU1}) is a textbook example of a quotient space construction~\cite{fulton}. What makes it less trivial here is the presence of additional structures on the manifold we compactify.
\end{rem}

One of the consequences of Def.~\ref{def:physicalstates} is that we cannot vary the parameter $p$ such that the probability of measurement of the bullet state is one (note that before we bounded $p$ the probability of measuring bullet had been one for $p=\pm1$). But this becomes more acceptable in the light of our earlier observation that the superqubit space is not a homogeneous space.

\begin{prop}\label{lem:compactindep}
The superqubit compactification Eq.~(\ref{eq:RtoU1}) is basis-independent.
\end{prop}
\begin{proof}
Up to now, we worked in a specific basis $\{\ket{0},\ket{1},\kbd\}$ but the compactification procedure should be independent on the basis. Let's see what happens if we transform a superqubit Eq.~(\ref{eq:superqubit}) into a rotated basis given by $\{Z\ket{0},Z\ket{1},Z\kbd\}$ where $Z=U(\a,\b)S(2p\eta)$  is an arbitrary $UOSP(1|2;\C\La_2)$ rotation. The  group action  followed by the change of Grassmann variables transforms $\ket{0}$ to $\ket{\psi}$ from Eq.~(\ref{eq:superqubitoriginal}). If compared to the $S(2p\eta)$ subgroup acting on $\ket{0}$ and followed by $p\mapsto-p$ one gets almost an identical state
$$
S(2p\eta)\ket{0}=\left(1+{p^2\over2}\eta\eta^\#\right)\ket{0}-p\eta\kbd.
$$
Only the action of $SU(2)$ is left out but that is confined to the even subspace and therefore is not relevant for the proof. So we can just study the effect of the rotated standard basis $\ket{0'}=S(2x\eta)\ket{0},\ket{1'}=S(2x\eta)\ket{1}$ and $\ket{\bu'}=S(2x\eta)\kbd$ where $x\in\bbR$. We rewrite the transformed superqubit as
$$
\ket{\psi'}=\left(1+{(p-x)^2\over2}\eta\eta^\#\right)\ket{0'}-(p-x)\eta\ket{\bu'}.
$$
We want this state to be a physical state according to Def.~\ref{def:physicalstates} and so we impose $|p'|\leq1/2$ where $p-x=p'$. But this is not enough and the argument now goes exactly as in the paragraph leading to Eq.~(\ref{eq:RtoU1}) -- the compactification in the new basis is achieved by the same prescription as Eq.~(\ref{eq:subst})
$
p'\mapsto {p'\over2\pi}-\bigl\lfloor{p'\over2\pi}\bigr\rfloor-{1\over2}.$
\end{proof}

\section{Bipartite superqubit states, the CHSH game and Tsirelson's bound}\label{sec:CHSH}

The most interesting results of quantum information theory are when bi- and multipartite states are used as resources in computational and communication protocols. Quantum correlations are the distinctive aspect of quantum physics and one of the consequences is that using  multipartite entangled quantum states one can perform significantly better compared to classical physics. Here we want to argue that multipartite entangled quantum states based on superqubits are even better resources than ordinary quantum states. But we face an obstacle. It is not  immediately obvious what is the Lie superalgebra one should study. Moreover, the representation theory of higher-dimensional Lie superalgebras is not straightforward~\cite{ritt,superLie}. We will follow a different path here. Using our definition of a super Hilbert space (Def.~\ref{def:hilbertspace}) we conjecture the existence of certain states for which there are good reasons to think that they are members of the carrier space of the Grassmann-valued group we would have obtained by studying higher orthosymplectic Lie superalgebras. One of such states is a tensor product of two superqubits. To construct it, let's utilize the transformed superqubits from Eq.~(\ref{eq:superqubitoriginal}) whose form leads to
\begin{align}\label{eq:tensorsuperqb}
  \ket{\psi}_A\ket{\psi}_B&=\Bigl(1+{p_A^2\over2}\eta_A\eta_A^\#\Bigr)\Bigl(1+{p_B^2\over2}\eta_B\eta_B^\#\Bigr)
  (\a\ket{0}+\b\ket{1})(\g\ket{0}+\d\ket{1})\nn\\
  &+p_B\eta_B\Bigl(1+{p_A^2\over2}\eta_A\eta_A^\#\Bigr)(\a\ket{0\bu}+\b\ket{1\bu})
  +p_A\eta_A\Bigl(1+{p_B^2\over2}\eta_B\eta_B^\#\Bigr)(\g\ket{\bu0}+\d\ket{\bu1})\nn\\
  &-p_Ap_B\eta_A\eta_B\ket{\bu\bu},
\end{align}
where $p_A,p_B\in\bbR$, $\a,\b,\g,\d\in\bbC$ and  $\eta_A,\eta_A^\#,\eta_B,\eta_B^\#\in\C\La_4$. As expected, the state $\ket{\bu\bu}$ does not appear accompanied by ordinary numbers as a consequence of  Lemma~\ref{lem:gradetensorprod}. Hence, we propose the second example of a pure two-superqubit state to be
\begin{align}\label{eq:2superqgeneral}
    {\Psi}_{AB}&=
        \Big(1+{X\over2}+{3\over8}X^2\Big)(a\ket{00}+b\ket{01}+c\ket{10}+d\ket{11})\nn\\
        &+p_B\eta_B\Big(1+{p_A^2\over2}\eta_A\eta_A^\#\Big)(\a\ket{0\bu}+\b\ket{1\bu})
        +p_A\eta_A\Big(1+{p_B^2\over2}\eta_B\eta_B^\#\Big)(\g\ket{\bu0}+\d\ket{\bu1})\nn\\
        &-p_Ap_B\eta_A\eta_B\ket{\bu\bu},
\end{align}
where $p_A,p_B\in\bbR$, $a,b,c,d,\a,\b,\g,\d\in\bbC$ such that $|a|^2+|b|^2+|c|^2+|d|^2=1$, $|\a|^2+|\b|^2=|\g|^2+|\d|^2=1$ and
$$
X=p_A^2\eta_A\eta_A^\#+p_B^2\eta_B\eta_B^\#+p_A^2p_B^2\eta_A\eta_A^\#\eta_B\eta_B^\#.
$$
This expression can factorized:
$$
1+{X\over2}+{3\over8}X^2=\Big(1+{p_A^2\over2}\eta_A\eta_A^\#\Big)\Big(1+{p_B^2\over2}\eta_B\eta_B^\#\Big).
$$
Note that the state $\Psi_{AB}$ contains an arbitrary two-qubit state.

Let's set $a=d=1/\sqrt{2}$ and $\b=\d=1$ and we obtain the state we are going to experiment with:
\begin{align}\label{eq:Upsilon}
    \Upsilon_{AB}(p_A,p_B)&=
        \Big(1+{p_A^2\over2}\eta_A\eta_A^\#\Big)\Big(1+{p_B^2\over2}\eta_B\eta_B^\#\Big){1\over\sqrt{2}}(\ket{00}+\ket{11})\nn\\
        &+p_B\eta_B\Big(1+{p_A^2\over2}\eta_A\eta_A^\#\Big)\ket{1\bu}
        +p_A\eta_A\Big(1+{p_B^2\over2}\eta_B\eta_B^\#\Big)\ket{\bu1}
        -p_Ap_B\eta_A\eta_B\ket{\bu\bu}.
\end{align}
We  claim that $\Upsilon_{AB}$ is at least as nonlocal as a maximally entangled (Bell) state. If we prepare any setup where a maximally entangled state is used in quantum information theory, utilize $\Upsilon_{AB}$ instead and ignore the bullet components ($p_A=p_B=0$) we will be able to perform as efficiently as with the Bell state itself. The question is now: Is $\Upsilon_{AB}$ able to perform better considering the super degrees of freedom? The best way to check is to reproduce the experiment that is a hallmark of nonlocality -- the coincidence measurement resulting in Bell's inequalities~\cite{bell}. There exists a sharp reformulation of Bell inequalities known as the CHSH game~\cite{CHSH} interpreting the measurement from the computer science point of view. Let us recapitulate the CHSH game. It is a so-called nonlocal game~\cite{CHSHgame} with three players: a referee who competes with two cooperating players Alice and Bob. The referee  chooses two bits $i\in\{0,1\}$  and $j\in\{0,1\}$ with probability $1/4$  and sends $i$ to Alice and $j$ to Bob such they are not aware of one another's bit value. Alice and Bob  each return a bit of communication (denoted $a$ and $b$, respectively) back to the referee. The condition for Alice and Bob to win the game is when the equation $ij=a\oplus b$ is satisfied for each round.

Alice and Bob cannot communicate during the game but they can establish their strategy beforehand. They also share a resource -- a physical system obeying the known laws of physics. The agreed strategy can be looked upon as a type of classical resource (classical correlations). In that case, the optimal strategy leads to the maximal probability of winning
$$
p_{win}^{class}={3\over4}.
$$
If they share quantum correlations the chances of winning are higher. Namely, a shared maximally entangled state $\Psi_{AB}=1/\sqrt{2}(\ket{00}+\ket{11})$ accompanied by an agreed measurement strategy leads to
$$
p_{win}^{quant}=\cos^2{\pi\over8}\simeq 0.8535.
$$
As a matter of fact, this is the maximal value that can be reached for the CHSH game using quantum-mechanical resources. It is known as Tsirelson's bound~\cite{tsirelson}. To achieve the bound they choose one of the following orthogonal measurement bases $\{\ket{0}_{iA},\ket{1}_{iA}\}$ and $\{\ket{0}_{jB},\ket{1}_{jB}\}$ rotated according to the value they receive from the referee $\{i,j\}\to\{\a_i,\b_i,\g_j,\d_j\}$ where
\begin{align*}
    \ket{0}_{iA}&=\a_i\ket{0}_A+\b_i\ket{1}_A\\
    \ket{0}_{jB}&=\g_j\ket{0}_B+\d_j\ket{1}_B
\end{align*}
and similarly for $\ket{1}_{iA(jB)}$. The amplitudes achieving Tsirelson's bound read
\begin{align*}
    (ij=00)&\to\big\{\a_0=1,\b_0=0,\g_0=\cos{\pi\over8},\d_0=\sin{\pi\over8}\big\}\\
    (ij=01)&\to\big\{\a_0=1,\b_0=0,\g_1=\cos{\pi\over8},\d_1=-\sin{\pi\over8}\big\}\\
    (ij=10)&\to\big\{\a_1={1\over\sqrt{2}},\b_1={1\over\sqrt{2}},\g_0=\cos{\pi\over8},\d_0=\sin{\pi\over8}\big\}\\
    (ij=11)&\to\big\{\a_1={1\over\sqrt{2}},\b_1={1\over\sqrt{2}},\g_1=\cos{\pi\over8},\d_1=-\sin{\pi\over8}\big\}.
\end{align*}

Up until now there has been no candidate among physical theories that could provide resources more nonlocal than a maximally entangled state. The only possibility is a nonlocal box (also called PR box)~\cite{PR} as a mathematical construct designed to reach the maximal winning probability $p^{PR}_{win}=1$. A nonlocal box is a hypothetical resource shared by Alice and Bob whose inputs are $i$ and $j$ and its highly nonlocal inner workings produce the values  $a$ and $b$ such that Alice and Bob always win.

If we want to test how well $\Upsilon_{AB}$ performs we have to adjust the rules of the CHSH game but at the same time we have to play exactly the same game as we play with a Bell state. A superqubit is formally a three-level system and so we merge the subspace spanned by $\ket{1}$ and $\kbd$. We set the rules such that Alice (Bob) announces the result $a=1$ ($b=1$) if the result of the measurement lies in this subspace and $a=0$ ($b=0$) if it was projected into $\ket{0}$. We define
\begin{equation}\label{eq:localsuperrotation}
    Z_{iA}\otimes Z_{jB}=S(2r_i\eta_A)U(\a_i,\b_i)\otimes S(2s_j\eta_B)U(\g_j,\d_j),
\end{equation}
where $\eta_A,\eth_A,\eta_B$ and $\eth_B$ are generators of the Grassmann algebra $\C\La_4$ and $r_i,s_j\in\bbR$ is chosen according to the bits $i$ and $j$ received from the referee. The local superunitary  transformation is a general rotation Eq.~(\ref{eq:UOSPgroup}) following Lemma~\ref{lem:Ssimplified} leading to $S(2p\eta)$ in Eq.~(\ref{eq:UOSPgroupS}).

The measurement will be performed on a shared bipartite entangled superqubit state $\Upsilon_{AB}$ rotated according to Eq.~(\ref{eq:localsuperrotation})
\begin{equation}\label{eq:rotatedGamma}
    {\Upsilon}_{iA,jB}=(Z_{iA}\otimes Z_{jB}){\Upsilon}_{AB}.
\end{equation}
Therefore the winning Grassmann-valued probability reads
\begin{equation}\label{eq:grassWINprob}
    p_{\G{win}}(\Upsilon_{AB})={1\over4}\sum_{ij\in\{00,01,10\}}
    \Big( p_{\G{00}}^{(ij)}+ p_{\G11}^{(ij)}+ p_{\G1\bu}^{(ij)}+ p_{\G\bu1}^{(ij)}+ p_{\G\bu\bu}^{(ij)}\Big)
    + p_{\G01}^{(11)}+ p_{\G10}^{(11)}+ p_{\G0\bu}^{(11)}+ p_{\G\bu0}^{(11)},
\end{equation}
where
\begin{align}\label{eq:probfunction}
     p^{(ij)}_{\G mn}&=(-)^{|m|\oplus|n|}\brk{m_{A}n_{B}}{\Upsilon_{iA,jB}}\big(\brk{m_{A}n_{B}}{\Upsilon_{iA,jB}}\big)^\#\nn\\
     &=(-)^{|m|\oplus|n|}\brk{m_{A}n_{B}}{\Upsilon_{iA,jB}}\!\brk{\Upsilon_{iA,jB}}{n_{B}m_{A}}
\end{align}
is the Grassmann-valued probability function introduced in Def.~\ref{def:grassmannVALUEDprob}. The letters $m$ and $n$ label the orthogonal basis states $\{\ket{m},\ket{n}\}=\{\ket{0},\ket{1},\kbd\}$. The phase factor in the first line  comes from Eq.~(\ref{eq:OrthoBasisMeasurement}) and the second line follows from Proposition~\ref{prop:formproperties}. Note that if we kept the basis order in $\ket{n_{B}m_{A}}$ to be $AB$ instead of $BA$ we would have to add an additional minus for the case when $n=m=\bu$, that is, both bases are odd. As a sanity check we can calculate the norm of $\Upsilon_{AB}$ to be
$$
\sum_{m,n=0,1,\bu}p^{(ij)}_{\G mn}=\sum_{m,n=0,1,\bu}(-)^{|m|\oplus|n|}\brk{m_{A}n_{B}}{\Upsilon_{iA,jB}}\!\brk{\Upsilon_{iA,jB}}{n_{B}m_{A}}=1
$$
for all choices of $i,j$.

\begin{thm}\label{thm:supernonlocal}
    The state $\Upsilon_{AB}$ from Eq.~(\ref{eq:Upsilon}) used as a resource in the CHSH game  with the restrictions on physical state (Def.~\ref{def:physicalstates})  crosses Tsirelson's bound reaching $p^{sqbit}_{win}\simeq 0.8647 $.
\end{thm}
\begin{rem}
Note the difference between $\Upsilon_{AB}$ and $\Ga_{AB}$ studied in~\cite{bbd}.
\end{rem}
\begin{proof}
We define
\begin{subequations}
\begin{align}\label{eq:maxprocedure}
         p_{win}=&\max_{\substack{p_A,p_B,r_i,s_j \\ \a_i,\b_i,\g_j,\d_j}}{p_{win}(\Upsilon_{AB})}\\
         \mbox{s.t.\hspace{4.254mm}} &|r_i|\leq1/2,|s_j|\leq1/2,\label{eq:maxprocedureii}\\
         \mbox{\hspace{4.254mm}} &|p_A|\leq1/{2},|p_B|\leq1/{2},\label{eq:maxprocedureiii}\\
         \mbox{\hspace{4.254mm}}
         &0\leq p^{(ij)}_{mn}\leq1,\hspace{5.959mm}\forall i,j,m,n,\label{eq:maxprocedureiv}
\end{align}
\end{subequations}
where $p_{win}(\Upsilon_{AB})$ is Eq.~(\ref{eq:grassWINprob}) after the modified Rogers norm from Def.~\ref{def:modrogers} has been used. The constraint in Eq.~(\ref{eq:maxprocedureii}) is Def.~\ref{def:physicalstates} applied on a tensor product of two superqubits $\psi(r_i)$ and $\vp(s_j)$ (cf. Eq.~(\ref{eq:pGRogers})).  The constraint in Eq.~(\ref{eq:maxprocedureiii}) follows from Def.~\ref{def:physicalstates} applied on $\Upsilon_{AB}$. It is surprisingly equivalent to the previous constraint since the transition probability factorizes
\begin{align*}
    p_\G(\Upsilon_{AB}(p_A,p_B),\Upsilon_{AB}(q_A,q_B))&=\Big(1+(p_A-q_A)^2\eta_A\eta_A^\#\Big)\Big(1+(p_B-q_B)^2\eta_B\eta_B^\#\Big)\\
    &\overset{\rm Def.~\ref{def:modrogers}}{\longrightarrow}
    \big(1-(p_A-q_A)^2\big)\big(1-(p_B-q_B)^2\big).
\end{align*}
The third line is a constraint that expresses our ignorance about how to get rid of negative probabilities for the measurement of $\Upsilon_{AB}$ in an arbitrary, locally superrotated,  basis. The simple procedure from Def.~\ref{def:physicalstates} followed by the compactification must be generalized. The reason is that there is no factorization happening for the amplitude
$$
\brk{m_An_B}{(Z_{iA}\otimes Z_{jB})\Upsilon_{AB}(p_A,p_B)}
$$
for an arbitrary rotation $Z_{iA}\otimes Z_{jB}$. These are the expressions forming the transition probability $p^{(ij)}_{\G mn}$ of a general projective measurement~Eq.~(\ref{eq:probfunction})  used for the calculation of the winning probability. So there does not seem to exist a sole condition on the $p_A,p_B$ parameters to get positive probabilities -- they are intertwined with the parameters $\a_i,\b_i,\g_j$ and $\d_j$ coming from the $SU(2)_A\otimes SU(2)_B$ subgroup.

Hence, we have no equivalent of Lemma~\ref{lem:compactindep} for single superqubits and Eqs.~(\ref{eq:maxprocedureii}) and~(\ref{eq:maxprocedureiii}) are not sufficient to guarantee the positivity of the transition probabilities. It must be enforced `manually' as in Eq.~(\ref{eq:maxprocedureiv}). This step is crude but if a consistent compactification is in principle possible even for two superqubits (that is an open question), it will lead to the same result -- a two-superqubit  Hilbert space that does not lead to negative transition probabilities. However, the two-superqubit manifold will likely be a non-trivial surface whose compactification might not be straightforward.

Note that we require all thirty six transition probabilities to lie between zero and one since the losing probabilities can be in principle measured  if Alice and Bob, for some reason,  decide to do so.

The overall expression for $p_{win}$ is complicated and its form is not really informative. The optimization has to be done numerically~\cite{yalmip} and gives us $p^{sqbit}_{win}\simeq 0.8647 $ with the following winning parameters:
\begin{align*}\label{eq:winningparameters}
    &p_A\simeq-1/2,\ p_B\simeq0\\
    &r_0\simeq-0.3450,\ s_0\simeq 0,\ r_1\simeq 0.3465,\ s_1\simeq 0\\
    &\a_0\simeq1.7768,\a_1\simeq -1.7749,\b_0\simeq \pi/2,\b_1\simeq -\pi/4.
\end{align*}
The optimization procedure leads to a non-convex program and so $p^{sqbit}_{win}\simeq 0.8647$ is not necessarily a global maximum.
\end{proof}

\section{Conclusions}\label{sec:concl}

In this work we studied superqubits -- supersymmetric quantum states based on a certain supersymmetric extension of quantum mechanics. The motivation for this work is to properly define the mathematical structures used in~\cite{superqubits,bbd} and offer a way  of getting rid of negative probabilities encountered in~\cite{bbd}. This has been achieved by a proposed method of compactification of the superqubit space thus resolving the problem for single superqubits. The problem remains open for multipartite superqubit states where there is a hope that the issue could be tackled in a similar way by considering higher-dimensional Lie superalgebras.

The paper contains two main parts followed by two appendices. In the first section  the algebraic properties of superqubits were studied in detail and a number of novel results were proven mainly for maps on supercommutative bimodules and related structures. This section builds upon the machinery of Lie superalgebras and superlinear algebra that has been extensively reviewed in Appendix~\ref{sec:AppBackground} followed by Appendix~\ref{sec:AppSLinearALg} with a number of practical rules for calculating with superqubits. Among several main results from the first section are the introduction of a super Hilbert space and the rules for obtaining real numbers from even Grassmann-valued probability functions based on the Rogers norm and Berezin integral. The prescription used here is novel and is more similar to the procedure of getting real numbers from Grassmann numbers introduced in~\cite{castell} than to~\cite{bbd}.

In the second section we ventured into the territory of multi-superqubit states and constructed certain bipartite superentangled states. One such state (a different one from the state used in~\cite{bbd}) was used as a nonlocal resource in a three-party  game known as the CHSH game. The game is a perspicuous reformulation of the CHSH inequalities from the quantum communication complexity theory point of view. The best performance quantum mechanics is capable of is when a maximally entangled state is used as a shared nonlocal resource  in the game between Alice and Bob. The maximum winning probability  is then $p^{quant}_{win}=\cos^2{\pi/8}\simeq0.8536$ which in terms of an expected value of an operator corresponds to so-called Tsirelson's bound~\cite{tsirelson}. It has been known, however, that quantum mechanics is not as nonlocal as it could have been. There exists a gap beyond Tsirelson's bound filled with hypothetical no-signalling theories (that is, theories not permitting superluminal communication) but more nonlocal than quantum mechanics. In~\cite{bbd} we reported crossing  Tsirelson's bound using a concrete physical model based on superqubits. Here, due to the introduced compactification procedure, we further limited the parameter space of superqubits while still being able to cross the bound. The maximal winning probability we found is lower compared to~\cite{bbd}: $p^{sqbit}_{win}\simeq 0.8647$.

This study leaves several questions unanswered. First of all, how else are superqubits different from quantum mechanics? Or, even more generally,  does this theory fit into the framework of general probabilistic theories studied recently by a number of authors~\cite{found1,found2,found3}? It might be of interest to see if all desirable axioms are satisfied and, if not, what the consequences would be. After all, the version of supersymmetric quantum mechanics we set out to explore possibly extends quantum mechanics even without crossing Tsirelson's bound. Even if Tsirelson's bound was not beaten we would still be left with states that are unlike ordinary quantum-mechanical states. This brings us to another question. How can we get rid of negative probabilities for bipartite, and possibly multipartite states? Negative probabilities are never used to calculate anything but the theory is still  incomplete since they can be reached by the group action followed by the modified Rogers norm. We believe that the compactification procedure introduced here can be generalized for multi-superqubit states. The answer how to achieve this goal certainly lies on the way to the proper definition of a Grassmann-valued group governing the evolution of multipartite superqubits. That is a research project on its own that we avoided and instead used a dirty way to get around the problem in Section~\ref{sec:CHSH} by using the insight from the theory of superqubits obtained in the first  section. Finally, in the previous work~\cite{bbd} we defined the modified Rogers norm as a way how to extract real numbers from even Grassmann number. This is by no means a unique procedure. It might be interesting to propose and study alternative prescriptions.

\appendix

\section{Background on Lie superalgebras and related structures}
\label{sec:AppBackground}

\numberwithin{equation}{section}

\begin{defi}\label{def:gradedvectorspace}
(i) Let $W=W^{[0]}\oplus W^{[1]}$ be a finite-dimensional $\bbZ_2$-graded linear vector space over $\bbK=\bbR,\bbC$, where the grading structure is isomorphic to $\bbZ_2$. When
\begin{align*}
  \dim{W^{[0]}} &= p \\
  \dim{W^{[1]}} &= q
\end{align*}
we will write $W=\bbK^{p|q}$ to indicate $\dim{K^{p|q}}=p+q$.

(ii) An element $w$ of the vector space is called homogeneous if $w\in W^{[i]}$.  The degree of a homogeneous element is defined $\deg{w}\equiv|w|=i\in\bbZ_2$. The zero (one) degree elements are called even (odd).

(iii) A set of homogeneous elements
\begin{equation}\label{eq:freebasis}
\{e_i,\dots,e_p,e_{p+1},\dots,e_{p+q}\},
\end{equation}
where we declare $|e_i|=0$ for $1\leq i\leq p$ and $|e_i|=1$ for $p+1\leq i\leq p+q$, is a basis for $\bbK^{p|q}$ if any $w\in \bbK^{p|q}$ can be uniquely  written as
$$
w=\sum_{i=1}^{p+q} \mu_i e_i
$$
and $\mu_i\in\bbK$. The set $\{e_i\}_{i=1}^{p+q}$ is called the standard basis if the basis elements are ordered as in Eq.~(\ref{eq:freebasis}).
\end{defi}
The  property that makes $\bbZ_2$-graded  vector spaces different from ordinary vector spaces is that the tensor product obeys the grading structure:
\begin{align*}
  (V\otimes W)^{[k]}=\bigoplus_{k=l\oplus m}V^{[l]}\otimes W^{[m]},
\end{align*}
where $\oplus$ stands for addition modulo two.

Other names for degree is parity (mostly in physics) or grade. Some authors insist on distinction between grade and degree. In the present work these two terms will be used interchangeably.
\begin{defi}\label{def:endos}
    Let $\bbK^{p|q}$ be a $\bbZ_2$-graded linear vector space. A linear operator $X\in{\rm End}(\bbK^{p|q})$ is said to be even (bosonic) if it is grade-preserving
    $$
    X(W^{[i]})=W^{[i]}
    $$
    and we write $|X|=0$. Similarly, $X$ is called odd (fermionic) if it is grade-reversing
    $$
    X(W^{[i]})=W^{[i\oplus1]}
    $$
    (\,$|X|=1$). The symbol $\oplus$ denotes addition  modulo two.
\end{defi}
We can readily illustrate the use of the standard basis from Def.~\ref{def:gradedvectorspace}. Any linear operator can be represented as a matrix of the block form~\cite{varadara,carmeli}
\begin{equation}\label{eq:ABDCmatrix}
X=
\begin{pmatrix}
  A & B \\
  C & D \\
\end{pmatrix},
\end{equation}
where $\dim{X}=p+q$. So, for example, the submatrix $C$ is a rectangular block with $q$ rows and $p$ columns. The matrix $X$ has entries in  $\bbK$. ${\rm End}(\bbK^{p|q})$ consists only of even or odd linear maps whose standard form reads
\begin{equation}\label{eq:ABDCmatrixEven}
X_e=
\begin{pmatrix}
  A & 0 \\
  0 & D \\
\end{pmatrix}
\end{equation}
for even maps and
\begin{equation}\label{eq:ABDCmatrixOdd}
X_o=
\begin{pmatrix}
  0 & B \\
  C & 0 \\
\end{pmatrix}
\end{equation}
for odd maps.

\begin{defi}\label{def:superalgebra}
(i) A $\bbZ_2$-graded ring $R$ is called a superalgebra  if it is furnished with a supercommutator (also called a graded commutator) $[,]:R\times R\to R$ defined as
\begin{equation}\label{eq:supercommutator}
[r,s]=rs-(-)^{|r||s|}sr
\end{equation}
valid for all $r,s\in R$.\\
(ii)
A superalgebra $R$ is called supercommutative  if
\begin{equation}\label{eq:supercomsuperal}
[r,s]=0
\end{equation}
holds for all $r,s\in R$.
\end{defi}
\begin{rem}
  A superalgebra from the above definition is formally not an algebra (it is trivially an algebra over the integers though~\cite{classic}). But this can be easily rectified. In particular, let there be a ring $R$ that is also a $\bbZ_2$-graded complex vector space such that
    \begin{equation}\label{eq:scalarMultiplic}
    \la(rs)=(\la r)s=r(\la s)
    \end{equation}
  is satisfied for all $r,s\in R$ and $\la\in\bbC$. Then $R$ is an algebra, namely, a $\bbZ_2$-graded algebra. From now on, when we say superalgebra we mean a $\bbZ_2$-graded algebra.
\end{rem}
If we adopted a more categorical approach to superalgebras~\cite{carmeli}, we could define the supercommutator without introducing rings and the related multiplication.
\begin{exa}
    A complex Grassmann algebra $\C\La_N$ of order $N$ is a traditional example of a supercommutative superalgebra.
    It is freely generated by $N$ anticommuting generators $\{\eta^i\}_{i=1}^N$ and it has a direct sum structure
    $$
    \C\La_N=\bigoplus_{k=0}^N\C\La_N^k,
    $$
    where $\dim{\C\La_N^k}=\binom{N}{k}$. The dimension of the  Grassmann algebra $\C\La_N$ is therefore $2^N$ and it contains a unit element in $\C\La_N^0\equiv\bbC$. Note that in this work we  consider only finite-dimensional Grassmann algebras. We will use $\C\La_{N,i}$ to denote an even ($i=0$) or odd ($i=1$) subspace of $\C\La_N$. Recall that the Grassmann algebra $\C\La_N$  is isomorphic to the exterior algebra $\wedge_N$. By linearity of the wedge product the supercommutator can  be extended to non-homogeneous elements of $\C\La_N$.
\end{exa}
\begin{defi}\label{def:supernumber}
    An arbitrary element $\zeta\in\C\La_N$ is called a supernumber and can be uniquely decomposed as $\zeta=\zeta_{e}+\zeta_{o}$ where $\zeta_{e}\in\C\La_{N,0}$ and $\zeta_{o}\in\C\La_{N,1}$. The general form of an even and odd supernumber reads
    \begin{align}\label{eq:supernumber}
        \zeta_{e}&=z_0+\sum_{k\in\bbN_e}\sum_{m=1}^{N\choose k}{1\over k!}z^{(m)}_I\eta^I
        =z_0+\sum_{k\in\bbN_e}\sum_{m=1}^{N\choose k}z^{(m)}\eta,\\
        \zeta_{o}&=\sum_{k\in\bbN_o}\sum_{m=1}^{N\choose k}{1\over k!}z^{(m)}_I\eta^I
        =\sum_{k\in\bbN_o}\sum_{m=1}^{N\choose k}z^{(m)}\eta,
    \end{align}
    where $z_0,z^{(m)}_I\in\mathbb{C}$, $\bbN_e(\bbN_o)$ is a subset of even (odd) integers $\bbN_e=\big\{2n;1\leq n\leq\lfloor{N\over2}\rfloor\big\}\ (\bbN_o=\big\{2n-1;1\leq n\leq\lfloor{N+1\over2}\rfloor\big\})$ and the multiindex $I$ is defined as $I=[i_1\dots i_{k}]$ where ${\eta^I=\eta^{i_1}\dots\eta^{i_k}}$ is a product of $k$ Grassmann generators.

    Furthermore, we will call even Grassmann numbers of grade zero and odd Grassmann numbers of grade one where the grade will be denoted by vertical lines: $|\zeta_{e}|\overset{\rm df}{=}0$ and $|\zeta_{o}|\overset{\rm df}{=}1$.
\end{defi}
Note that we sum over $I$ but since $z_I^{(m)}$ is a completely antisymmetric tensor we set $I=1\dots k$ and so $z^{(m)}=z^{(m)}_I$ and $\eta=\eta^I$ on the RHS of the above equations.

\begin{defi}\label{def:Liesuperalgebras}\cite{superLie,superLie1,kostant,musson}
     A finite-dimensional a $\bbZ_2$-graded algebra $R$ is called a Lie superalgebra if it is equipped with a bilinear non-associative product $[,]:R\times R\mapsto R$ satisfying
     \begin{align}\label{eq:LIEsuperalgebra}
         [r,s]&=-(-1)^{|r||s|}[s,r],\\
         0&=(-1)^{|r||t|}[r,[s,t]]+(-1)^{|s||r|}[s,[t,r]]+(-1)^{|t||s|}[t,[r,s]]
     \end{align}
    for all $r,s,t\in R$.
\end{defi}
One can verify that the graded commutator Eq.~(\ref{eq:supercommutator}) satisfies the above conditions.
\begin{exa}
    The general linear Lie superalgebra $gl(p|q;\bbK)$ is simply ${\rm End}(\bbK^{p|q})$ as introduced in Def.~\ref{def:endos}~\cite{carmeli}. The graded Lie product from Def.~\ref{def:Liesuperalgebras} is defined as $[X,Y]=XY-(-)^{|X||Y|}YX$ with the usual matrix multiplication implied.
\end{exa}
\begin{defi}\label{def:ST}
  Let $X\in{\rm End}(\bbK^{p|q})$ be written in the standard basis Eq.~(\ref{eq:freebasis}). The supertranspose of $X$ is defined as
\begin{equation}\label{eq:ST}
  X^{ST}\df
  \begin{pmatrix}
  A^T & (-)^{|X|}C^T \\
  -(-)^{|X|}B^T & D^T \\
\end{pmatrix},
\end{equation}
where $M^T$ denotes the transposition of a matrix $M$ in the standard basis.
\end{defi}
\begin{rem}
  Equivalently, we may write the component version of the supertranspose definition:
  $$
  x_{ji}^{ST}=x_{ij}(-)^{|X|(|j|\oplus|i|)\oplus|j|(|i|\oplus|j|)}.
  $$
  The standard basis convention dictates $|i|=0$ for $i\leq p$ and $|j|=0$ for $j\leq q$.
\end{rem}
This ad hoc looking definition is a special case of a  definition for more general object called supermatrices. We will get to them in a moment but for the sake of clarity it seems advantageous to first illustrate the concept on ${\rm End}(\bbK^{p|q})$. It follows from the Def.~\ref{def:ST} and Eqs.~(\ref{eq:ABDCmatrixEven}) and~(\ref{eq:ABDCmatrixOdd}) that
\begin{align}
  X_e^{ST}= & \begin{pmatrix}
  A^T & 0 \\
  0 & D^T \\
\end{pmatrix}, \\
  X_o^{ST}= &  \begin{pmatrix}
  0 & -C^T \\
  B^T & 0 \\
\end{pmatrix}.
\end{align}
We pinpoint two interesting properties of the supertranspose~\cite{manin,buch,varadara}:
\begin{align}\label{eq:STcomposition}
 (XY)^{ST} &= (-)^{|X||Y|}Y^{ST}X^{ST}, \\
 \left({X^{ST}}\right)^{ST} &=
 \begin{pmatrix}
  A & -B \\
  -C & D \\
\end{pmatrix}.\label{eq:STtwice}
\end{align}
Another reason to introduce the supertranspose at this point is the following important Lie supersubalgebra~\cite{superLie1}:
\begin{defi}\label{def:ospalgebra}
The real orthosymplectic Lie superalgebra $osp(p|q;\bbR)$ is defined as
    $$
    osp(p|q;\bbR)\df\{X\in gl(p|q;\bbR)|X^{ST}H+(-)^{|X|}HX=0\}.
    $$
    The matrix
    $$
    H=\begin{pmatrix}
       H_1 & 0  \\
       0 & H_2
     \end{pmatrix},
    $$
    represents a non-degenerate bilinear form where $H_1$ is a  symmetric matrix and $H_2$ is a skew-symmetric matrix.
\end{defi}
The algebra is a $\bbZ_2$-graded vector space where $\dim{H_1}=p$ and $\dim{H_2}=q$. From the matrix representation of the bilinear form follows that the subspaces spanned by even and odd basis elements are orthogonal with respect to it. We may rewrite the condition for a matrix $X$ to be in $osp(p|q;\bbR)$ as
\begin{equation}\label{eq:orthosymplcondition}
  A^TH_1+H_1A=D^TH_2+H_2D=B^TH_1-H_2C=0.
\end{equation}
Putting $H_1$ and $H_2$ in the standard form where $H_1$ is a $p$-dimensional unit  matrix and $H_2$ is $q\times q$ symplectic matrix (the form represented by $H$ is non-degenerate so $q$ is even) explains the name  orthosymplectic: the even subspace (even endomorphisms in the sense of Def.~\ref{def:endos}) is a direct sum of two Lie algebras bearing the same name. The odd subspace does not form an algebra.

\subsubsection*{Supermatrices}

The origin of matrices in linear algebra and the related operations on them (such as transpose)  revolves around the concept of duality of vector spaces (for a clear exposition see~\cite{classic}). We only briefly recall that every finite-dimensional vector space $W$ has a dual $W^*$ whose elements are linear forms $\om\in W^*$. The action of a linear form $\om:W\mapsto\bbK$ is usually written as $\om(w)$ where $w\in W$. By choosing a basis $\{b_i\}$ in $W$, this expression defines the dual basis $\{\b_i\}\in W^*$ by setting $\b_k(b_l)=\d_{kl}$. The spaces are  isomorphic but to make it basis-independent, an assistance of a non-degenerate bilinear form $F_{U,W}:U\times W\mapsto\bbK$ is required. For all $w\in W$ we obtain a linear form $F_{U,W}(_-,w):U\mapsto\bbK$ and so the isomorphism of $W$ and $U^*$ is given by the identification $w\mapsto F_{U,W}(_-,w)$.

The transpose operation  plays a fundamental role in linear algebra and appears in two slightly different contexts~\cite{classic}. First, a linear transformation $g:W\mapsto V$ defines a dual map $g^*:V^*\mapsto W^*$ by $g^*(\nu)\df \nu\circ g$ where $\nu\in V^*$ is a linear form and so\footnote{We recognize a pullback of $\nu$ along $g$~\cite{qft}.} $\nu:V\mapsto\bbK$. If $G$ is a matrix of $g$ with respect to the bases of $W$ and $V$ then $G^T$ is the matrix form of the dual map $g^*$ written in the corresponding dual bases of $V^*$ and $W^*$. The second occurrence of the transpose operation is after an additional structure has been introduced to the vector spaces $W$ and $V$, namely a non-degenerate bilinear form $F_{V}\equiv F_{V,V}$ and $F_W\equiv F_{W,W}$. It is at this point when we can employ the isomorphism $V\mapsto V^*$ and $W\mapsto W^*$ provided by the  identification mentioned in the previous paragraph. Let $g:W\mapsto V$ and $h:V\mapsto W$ be  linear maps (morphisms). Then $h$ is called the {\em adjoint} if it satisfies
$$
F_{V}(v,g(w))=F_{W}(h(v),w)
$$
for all $v\in V,w\in W$. It turns out that if $G$ is a matrix representing the map $g$ (written in the basis orthogonal with respect to $F_V$) than  the representing matrix $H$ of $h$ is just $G^T$. So taking the adjoint is formally the same thing as the transpose operation but one has to be aware of subtle differences important in a more general case of $\bbZ_2$-graded modules.

If we further relax the requirement of a field in the definition of  a vector space and let it be a non-commutative ring $R$, we obtain the definition of a left or right $R$-module and the correspondingly generalized notion of duality for modules~\cite{classic}. Note that even though a module is a more general structure than a vector space, it is often said that an $R$-module {\em is} a vector space over~$R$. The module axioms~\cite{classic} justify this type of language used mainly in the literature on supersymmetry~\cite{carmeli}.  In reality, modules over rings are much more general structures than vector spaces. But $R$-modules studied in supersymmetry are special -- they are {\em free} which is equivalent to saying that they admit a basis~\cite{classic}. Crucially, this basis can be chosen as the standard (canonical) basis  in linear algebra. This is precisely the choice of homogeneous elements in Eq.~(\ref{eq:freebasis}) with an addition of $\bbZ_2$-grading for the purposes of supersymmetry.

Following~\cite{carmeli,varadara,manin,berezin_book,buch}, it is possible to generalize this construction  in two principal directions. In the $\bbZ_2$-graded case the starting point is a  vector space $\bbK^{p|q}$. The first upgrade is to promote it to a supermodule. Note that in the spirit of the remark below Def.~\ref{def:superalgebra} we will be using the word superalgebra for a $\bbZ_2$-graded ring with an added compatible multiplication from a given field (see Eq.~(\ref{eq:scalarMultiplic})).
\begin{defi}\label{def:supermodule}
    Let $R$ be a supercommutative superalgebra (Def.~\ref{def:superalgebra}). The left $R$-supermodule is a $\bbZ_2$-graded vector space $W$ endowed with a left multiplication $R\times W\mapsto W$. Similarly, for the right $R$-supermodule we have a right multiplication $W\times R\mapsto W$.
\end{defi}
It is known~\cite{carmeli,varadara} that if the superalgebra $R$ is supercommutative, both multiplications are related by
\begin{equation}\label{eq:LefRightMult}
    wr=(-)^{|r||w|}rw,
\end{equation}
for all (homogeneous) $w\in W$ and $r\in R$. Then the resulting object is called (super)$R$-bimodule. In this work, the supercommutative superalgebra $R$ will always be the Grassmann algebra $\C\La_N$ of order $N$. We will occasionally denote  such  $R$-bimodules as $\La^{p|q}$.

Now we can proceed as in Def.~\ref{def:gradedvectorspace} and by using the basis from Eq.~(\ref{eq:freebasis}) we write down an element $w$ of the bimodule $\La^{p|q}$ as
\begin{equation}\label{eq:RbimoduleRight}
w=\sum_{i=1}^{p+q}e_i\zeta^r_i=\sum_{i=1}^{p+q}(-)^{|i||\zeta^r_i|}\zeta^r_ie_i,
\end{equation}
where $\zeta^r_i\in\C\La_N$ are the right components. It is customary to write the right components as a column vector~\cite{manin} (see~\cite{classic} for non-graded modules). The dual of the right $R$-module is a left $R$-module $\La^{*p|q}$. Similarly to the right $R$-module one can show that for any $\om\in\La^{*p|q}$ defined as $\om:w\mapsto R$ we obtain
\begin{equation}\label{eq:RbimoduleLeft}
\om=\sum_{i=1}^{p+q}\zeta^l_i\e_i=\sum_{i=1}^{p+q}(-)^{|i||\zeta^l_i|}\e_i\zeta^l_i,
\end{equation}
where $\zeta^l_i\in\C\La_N$ are the left components and $\{\e_i\}$ is the dual basis: $\e_i(e_j)=\d_{ij}$. The left components are written as rows and this convention has its origin precisely in the fact that in both graded and non-graded case, the left $R$-module (as a linear form) acts on the elements of the right $R$-module. This can be displayed as a row vector of the left coordinates multiplying a column vector of the right coordinates with the result in $R$. However, if we compare Eqs.~(\ref{eq:RbimoduleRight}) and (\ref{eq:RbimoduleLeft}) we can see that unlike the non-graded case (and for $R=\bbK$), the ordinary transpose operation does not achieve the swap of the left and right coordinates because of the signs that got in the way.

To proceed we note that relative to the standard basis, any linear map $\tau:\La^{p|q}\mapsto\La^{s|t}$ can be presented as a {\em supermatrix}\,:
\begin{equation}\label{eq:ABDCmatrixAgain}
S=
\begin{pmatrix}
  A & B \\
  C & D \\
\end{pmatrix}.
\end{equation}
Supermatrices have the block structure similar to ${\rm End}(\bbK^{p|q})$ Eq.~(\ref{eq:ABDCmatrix}) but the entries are now Grassmann numbers since $R=\C\La_N$. The supermatrix representing a morphism $\tau$ acts on a column vector (as they are elements of the right $R$-bimodule) from the left. Similarly, the supermatrix representing the action of the dual map $\tau^*:\La^{*s|t}\mapsto\La^{*p|q}$ acts on row elements of the left $R$-bimodule from the right. We  define a supermatrix $S$ to be even if the corresponding map preserves the parity and odd if it reverses it. In the former case, the entries of $A$ and $D$ are even Grassmann and the  entries of $C$ and $B$ are odd Grassmann numbers. For $S$  odd, the parity of entries of its  subblocks is swapped. Even and odd supermatrices are called homogeneous (sometimes called pure).
\begin{defi}\label{def:SupermatrixSet}
  The set of homogeneous supermatrices of dimension $(s+t)\times(p+q)$ with entries in $\C\La_N$ is denoted by $\M(s|t,p|q;\C\La_N)$. When $s=p$ and $t=q$ we will write $\MM$.
\end{defi}
\begin{rem}
  For our purposes, $\MM$ is a set but it is straightforward to promote it to an associative algebra with the usual matrix multiplication and further  define an associated Lie bracket from Def.~\ref{def:Liesuperalgebras} making it into a Lie superalgebra~\cite{carmeli}.
\end{rem}
Note that if  $R=\bbK$, an even supermatrix $S\in\MM$ becomes $X_e$, Eq.~(\ref{eq:ABDCmatrixEven}), and an odd supermatrix becomes $X_o$ (Eq.~(\ref{eq:ABDCmatrixOdd})).

For an $R$-bimodule  morphism $\tau:\La^{p|q}\mapsto\La^{s|t}$ there exists~\cite{manin,varadara}  its dual $\tau^*:\La^{*s|t}\mapsto\La^{*p|q}$ satisfying
\begin{equation}\label{eq:superpullback}
  (\tau^*(\om))(w)=(-)^{|\tau^*||\om|}(\om)(\tau(w)),
\end{equation}
where $w\in\La^{p|q}$ and $\om\in\La^{*s|t}$. The definition of the dual supermodule action generalizes the linear algebra construction sketched at the beginning of this subsection. If the matrix form of $\tau$ is a supermatrix $S$  with respect to the bases of $\La^{p|q}$ and $\La^{s|t}$ (Eq.~(\ref{eq:ABDCmatrix})) then the supermatrix representing the dual map $\tau^*$ written with respect to the bases of $\La^{*s|t}$ and $\La^{*p|q}$ is $T=S^{ST}$. $ST$ stands for the supertranspose and the definition coincides with Eq.~(\ref{eq:ST}) (assuming the standard basis):
\begin{equation}\label{eq:STforS}
    S=\begin{pmatrix}
      A & B \\
      C & D \\
      \end{pmatrix}\overset{ST}{\to}
      \begin{pmatrix}
      A^T & (-)^{|S|}C^T \\
      -(-)^{|S|}B^T & D^T \\
  \end{pmatrix}.
\end{equation}
\begin{defi}\label{def:superRowCol}
  (i) Let $z_{row}$ be a row supermatrix whose components are the left coordinates  $z_{row}(i)=\zeta^l_i$ of $\om\in\La^{*p|q}$. Its supertranspose is a  column supervector  $z_{col}={z}_{row}^{ST}$ where ${z}_{col}(i)\df(-)^{|i||\zeta^l_i|}\zeta^l_i$.\\
  (ii) Let $z_{col}$ be a column supervector whose components are the right coordinates  $z_{col}(i)=\zeta^r_i$ of $w\in\La^{p|q}$. Its supertranspose is a row supermatrix  $z_{row}={z}_{col}^{ST}$ where ${z}_{row}(i)\df(-)^{|i|(|\zeta^r_i|\oplus1)}\zeta^r_i$.
\end{defi}
It may seem a bit odd to use the same symbol $ST$ for an operation on rows/columns and supermatrices. For supermatrices we know that they represent  supermodule morphisms and the supertranspose gives us the dual  morphism. But the rows and columns of coordinates do not have any such interpretation. One option is to consider rows and columns as simple supermatrices and then we have to make sure that both  operations (that is, $ST$ from Eq.~(\ref{eq:STforS}) and the one brought in Def.~\ref{def:superRowCol}) are consistent so that we can both call them supertranspose.

But at first sight, it is not obvious what is going on. To clarify, we look for the inspiration in the non-graded case. If $\{e_i\}$ is a free basis of the vector space $V$ (a module over $\bbR$) then the components of  $v\in V$, where $v=\sum_i v_ie_i$, are represented by a column vector and there is no need to distinguish between left and right coordinates (so we wrote them on the left). An element of $f\in V^*$ of the space dual to $V$ written with respect to the basis $\{\e_i\}$ dual to $\{e_i\}$ reads $f=\sum_i f_i\e_i$ but its components $\{f_i\}$ are also represented by a column vector. On the other hand, the form $f$ written in the basis $\{e_i\}$ is represented as a row vector which is the transpose of the original row vector. But to be able to do this, we had to identify the spaces $V$ and $V^*$ through a non-degenerate bilinear form. In other words, our original vector space $V$ already has some additional structure enabling us to `multiply' columns by rows (this is the ordinary dot product yielding a real number).

The same discussion carries over to the super scenario  where of course one has to be careful to distinguish the left and right multiplication of the $R$-bimodule and take into account the properties of the underlying ring $R$. In the supersymmetric case we have $R=\C\La_N$  and the result is the modified transpose -- the supertranspose with all its different properties compared to the ordinary transpose. For more on this topic, see the beginning of the next subsection.

Having the previous paragraph in mind, let's go back to Def.~\ref{def:superRowCol}. The first part of the definition is suggested by comparing the coordinates in Eqs.~(\ref{eq:RbimoduleRight}) and (\ref{eq:RbimoduleLeft}) leading to $\zeta^r_i=(-)^{|i||\zeta^l_i|}\zeta^l_i$ as has been defined. But there is an ambiguity. The other possibility is $\zeta^l_i=(-)^{|i||\zeta^r_i|}\zeta^r_i$. The difference ultimately boils dow to the fact that the supertranspose is not an involution~\cite{manin,varadara} but an operation of order 4:
\begin{equation}\label{eq:STchain}
  S=\begin{pmatrix}
  A & B \\
  C & D \\
  \end{pmatrix}\overset{ST}{\to}
  \begin{pmatrix}
  A^T & (-)^{|S|}C^T \\
  -(-)^{|S|}B^T & D^T \\
  \end{pmatrix}\overset{ST}{\to}
  \begin{pmatrix}
  A & -B \\
  -C & D \\
  \end{pmatrix}\overset{ST}{\to}
  \begin{pmatrix}
  A^T & -(-)^{|S|}C^T \\
  (-)^{|S|}B^T & D^T \\
  \end{pmatrix}\overset{ST}{\to}
  S.
\end{equation}
The last sentence will be clarified after the next example.
\begin{exa}\label{exa:STonVectors}
    Let's verify on a simple example that the supertranspose action on a supermatrix is consistent with a supermatrix acting on a column vector of coordinatates as defined in the Def.~\ref{def:superRowCol}. Let $R=\C\La_N$ with $N$ high enough such that two identical Grassmann numbers do not meet upon multiplication (otherwise it may become trivial) and $\dim{W}=1+1$. The supermatrix  then reads
    \begin{equation}\label{eq:StwoBYtwo}
      S=\begin{pmatrix}
      a & b \\
      c & d \\
      \end{pmatrix},
    \end{equation}
  where $a,b,c,d\in\C\La_N$ such that $S$ is pure (even or odd). It will be acted upon a row supermatrix $z$ which we set to be
  $$
  z=\begin{pmatrix}
  1 ,& \eta  \\
  \end{pmatrix}
  $$
  for $z$ even and
  $$
  z=\begin{pmatrix}
  \eta,  & 1  \\
  \end{pmatrix}
  $$
  for $z$ odd and $\eta\in\C\La_{N,1}$. These particular choices do not weaken the generality of the conclusion. We calculate $z'=(zS)^{ST}$ and show that it coincides with
  $z''=(-)^{|S||z|}S^{ST}z^{ST}$ for all four possibilities: $|S|=0,1$ and $|z|=0,1$:
  \begin{itemize}
    \item $|S|=0,|z|=0$\\
    $$
    z'=
    \begin{pmatrix}
      1, & \eta  \\
    \end{pmatrix}
    \begin{pmatrix}
      a & b \\
      c & d \\
    \end{pmatrix}=
    \begin{pmatrix}
      a+\eta c, & b+\eta d  \\
    \end{pmatrix}\overset{ST}{\to}
    \begin{pmatrix}
     a+\eta c \\
     -b-\eta d  \\
    \end{pmatrix},
    $$
    $$
    z''=
    \begin{pmatrix}
      a & c \\
      -b & d \\
    \end{pmatrix}
    \begin{pmatrix}
    1 \\
    -\eta  \\
    \end{pmatrix}=
    \begin{pmatrix}
    a-c\eta \\
    -b-d\eta  \\
    \end{pmatrix}=
    \begin{pmatrix}
     a+\eta c \\
     -b-\eta d  \\
    \end{pmatrix}\equiv
    z'.
    $$

\item $|S|=0,|z|=1$\\
    $$
    z'=
    \begin{pmatrix}
      \eta, & 1  \\
    \end{pmatrix}
    \begin{pmatrix}
      a & b \\
      c & d \\
    \end{pmatrix}=
    \begin{pmatrix}
      \eta a+ c, & \eta b+ d  \\
    \end{pmatrix}\overset{ST}{\to}
    \begin{pmatrix}
     \eta a+ c \\
     \eta b+ d  \\
    \end{pmatrix},
    $$
    $$
    z''=
    \begin{pmatrix}
      a & c \\
      -b & d \\
    \end{pmatrix}
    \begin{pmatrix}
    \eta \\
    1  \\
    \end{pmatrix}=
    \begin{pmatrix}
    a\eta+c \\
    -b\eta+d  \\
    \end{pmatrix}=
    \begin{pmatrix}
     \eta a+ c \\
     \eta b+ d  \\
    \end{pmatrix}\equiv
    z'.
    $$

\item $|S|=1,|z|=0$\\
    $$
    z'=
    \begin{pmatrix}
      1, & \eta  \\
    \end{pmatrix}
    \begin{pmatrix}
      a & b \\
      c & d \\
    \end{pmatrix}=
    \begin{pmatrix}
       a+ \eta c, &  b+ \eta d  \\
    \end{pmatrix}\overset{ST}{\to}
    \begin{pmatrix}
      a+ \eta c \\
      b+ \eta d  \\
    \end{pmatrix},
    $$
    $$
    z''=
    \begin{pmatrix}
      a & -c \\
      b & d \\
    \end{pmatrix}
    \begin{pmatrix}
    1 \\
    -\eta  \\
    \end{pmatrix}=
    \begin{pmatrix}
    a+c \eta \\
    b-d \eta  \\
    \end{pmatrix}=
    \begin{pmatrix}
      a+ \eta c \\
      b+ \eta d  \\
    \end{pmatrix}\equiv
    z'.
    $$

\item $|S|=1,|z|=1$\\
    $$
    z'=
    \begin{pmatrix}
      \eta, &  1 \\
    \end{pmatrix}
    \begin{pmatrix}
      a & b \\
      c & d \\
    \end{pmatrix}=
    \begin{pmatrix}
       \eta a+  c, &  \eta b+  d  \\
    \end{pmatrix}\overset{ST}{\to}
    \begin{pmatrix}
      \eta a+  c \\
      -\eta b- d  \\
    \end{pmatrix},
    $$
    $$
    z''=
    -\begin{pmatrix}
      a & -c \\
      b & d \\
    \end{pmatrix}
    \begin{pmatrix}
    \eta \\
    1  \\
    \end{pmatrix}=
    -\begin{pmatrix}
    a\eta-c  \\
    b\eta+d  \\
    \end{pmatrix}=
    \begin{pmatrix}
      \eta a+  c \\
      -\eta b- d  \\
    \end{pmatrix}\equiv
    z'.
    $$

   \end{itemize}
\end{exa}
Encouraged by the previous example, it seems that the supertranspose operations on supematrices and supervectors are compatible exactly as in the non-graded case. Indeed, a row supermatrix is considered to be a square supermatrix $S\in\MM$ of size $(1+0)\times(p+q)$ (the uppermost row of $S$) and a column supervector is a supermatrix of size $(p+q)\times(1+0)$ (the leftmost column of $S$). When the row is supertransposed, we use the first part of Def.~\ref{def:superRowCol} and it coincides with the Def.~\ref{def:ST} applied to supermatrices. It also provides the definition for the supertranspose of a column vector $z_{row}=z_{col}^{ST}$ where $z_{row}(i)\df-(-)^{|i||\zeta^r_i|}\zeta^r_i$~\cite{manin}. We have the following chain of how the supertranpose transforms an even and odd supervector (let's take $z$ from the previous example):
\begin{align}\label{eq:STactsOnVectors1}
    \begin{pmatrix}
      1, &  \eta \\
    \end{pmatrix} &\overset{ST}{\to}
    \begin{pmatrix}
    1 \\
    -\eta  \\
    \end{pmatrix} \overset{ST}{\to}
    \begin{pmatrix}
      1, &  -\eta \\
    \end{pmatrix} \overset{ST}{\to}
    \begin{pmatrix}
    1 \\
    \eta  \\
    \end{pmatrix} \overset{ST}{\to}
    \begin{pmatrix}
      1, &  \eta \\
    \end{pmatrix},
    \\
    \begin{pmatrix}
      \eta, &  1 \\
    \end{pmatrix} &\overset{ST}{\to}
    \begin{pmatrix}
    \eta \\
    1  \\
    \end{pmatrix}\overset{ST}{\to}
    \begin{pmatrix}
      \eta, &  -1 \\
    \end{pmatrix} \overset{ST}{\to}
    \begin{pmatrix}
    \eta \\
    -1  \\
    \end{pmatrix} \overset{ST}{\to}
    \begin{pmatrix}
      \eta, &  1 \\
    \end{pmatrix}.
\label{eq:STactsOnVectors2}
\end{align}
Let's get back to the second part of Def.~\ref{def:superRowCol}. In reality, there are two {\em equivalent} definitions of the supertranspose. It is either Def.~\ref{def:ST} leading to the chain Eq.~(\ref{eq:STchain}) we are using here or, alternatively,
\begin{equation}\label{eq:STalt}
  S^{ST_{alt}}\df
  \begin{pmatrix}
  A^T & -(-)^{|S|}C^T \\
  (-)^{|S|}B^T & D^T \\
\end{pmatrix}.
\end{equation}
If we closely look at~Eq.~(\ref{eq:STchain}) then the new definition corresponds to reversing the arrows of the $ST$ action. And indeed, the second part of Def.~\ref{def:superRowCol} would be an alternative rule for the column supermatrix supertranpose in this case (cf. Eqs.~(\ref{eq:STactsOnVectors1}) and~(\ref{eq:STactsOnVectors2}) after reversing the arrows).

\subsubsection*{(Super)kets and bras}
\label{page:superkets}

Let us recall what kets and bras represent in quantum mechanics. Let $V$ be a vector space equipped with a non-degenerate Hermitian and positive semidefinite form $F_V:V\times V\mapsto\bbC$. If $\dim{V}=n<\infty$ the representation of the form is the $n$-dimensional unit matrix and the space $V$ can be called a Hilbert space. Any $v\in V$ is denoted as a ket $\ket{v}$ and the Hermitian form $F_V$ is in quantum mechanics written as $\brk{_-}{_-}$. A bra $\bra{u}$ is an element of the space $V^*$ dual to $V$ precisely because of the identification $u\mapsto\brk{u}{_-}$ provided by the Hermitian form $F_V$. Indeed, $\brk{u}{_-}:V\mapsto\bbC$ so it is a linear form whose shorthand notation is~$\bra{u}$. So there is a double-meaning to the symbol $\brk{_-}{_-}$: As we said, it is the same thing as $F_V$. But $\bra{u}$ also acts on $\ket{v}$ as $\bra{u}(\ket{v})$ -- a clumsy notation that is avoided by setting $\bra{u}\left(\ket{v}\right)\equiv\brk{u}{v}$. This overlaps with the primary meaning of $\brk{_-}{_-}$ but, fortunately, it does not cause troubles due to the aforementioned identification $V\mapsto V^*$.

The generalization of kets and bras to the  supersymmetric case is in many aspects similar. We can again assume the existence of a bilinear, non-degenerate form and identify the $R$-bimodule $\La^{p|q}$ with its dual. But we omitted the adjectives Hermitian and positive semidefinite for the form! We can assume the form to be Hermitian  if we restrict our attention to $\bbC^{p|q}\subset\MM$ and look for the inspiration to~\cite{snr1}:
\begin{defi}\label{def:gradeadjoint}
    Let $A,B\in{\rm End}(\bbC^{p|q})$ be homogeneous and $\brk{_-}{_-}:\bbC^{p|q}\times\bbC^{p|q}\mapsto\bbC$ be a non-degenerate Hermitian form such that the even
    and odd subspace are orthogonal with respect to it. We define a mapping $\ddg:{\rm End}(\bbC^{p|q})\mapsto {\rm End}(\bbC^{p|q})$ called the grade adjoint
    satisfying
    \begin{equation}\label{eq:gradedDEF}
      \brk{Az}{s}=(-)^{|A||z|}\brk{z}{A^\ddg s},
    \end{equation}
    valid for all homogeneous $z,s\in \bbC^{p|q}$. Let the grade adjoint satisfy the following properties:
    \begin{align}
      (aA+bB)^\ddg&=\bar{a}A^\ddg+\bar{b}B^\ddg,\label{eq:gradedadjoint}\\
      (AB)^\ddg&=(-)^{|A||B|}B^\ddg A^\ddg\label{eq:gradedadjointii},\\
      (A^\ddg)^\ddg&=(-)^{|A|}A\label{eq:gradedadjointiii},
    \end{align}
    where $a,b\in\bbC$ and the bar denotes complex conjugation.
\end{defi}
Note that for $q=0$ we get $A^\ddg=A^\dg$ (all operators are even), the dagger becomes the usual quantum-mechanical adjoint and we can impose positive semidefiniteness on the bilinear form. Apart from this trivial example of an operation satisfying the above axioms, we already have a less trivial candidate for the double dagger if $q\neq0$: $\ddg\df \overline{ST}$ (the bar denotes  complex conjugation and it commutes with $ST$). Modifying the example on page~\pageref{exa:STonVectors} by setting $S=A\in\bbC^{1|1}$ in Eq.~(\ref{eq:StwoBYtwo}), we have $A$ even (zeros on the off-diagonal) or odd (zeros on the diagonal) with the non-zero entries in $\bbC$ and  $z=(z_1,0)$ for $|z|=0$ and $z=(0,z_2)$ for $|z|=1$ assuming $z_1,z_2\in\bbC$. Then we can show that $(Az)^{ST}=(-)^{|A||z|}z^{ST}A^{ST}$ holds. The antilinearity, Eq.~(\ref{eq:gradedadjoint}), is immediately satisfied and requirements~(\ref{eq:gradedadjointii}) and~(\ref{eq:gradedadjointiii}) follow from Eqs.~(\ref{eq:STcomposition}) and (\ref{eq:STtwice}).
\begin{rem}
  The reason why we avoided positive semidefiniteness in the above definition is precisely for the case where $q\neq0$. Than the tensor product of two vectors from the $\bbZ_2$-graded vector space $\bbC^{p|q}$ whose norms are positive does not need to be positive. This is an observation already made in Ref.~\cite{snr1} and an explicit example is the double bullet state Eq.~(\ref{eq:gradestarrepproduct}). 
\end{rem}
\begin{rem}[{\bf Important}]\label{rem:superkets}
  Now we can address the problem of the super version of kets and bras. They simply denotes elements of $\bbC^{p|q}$. Later, they will be generalized in the context of Theorem~\ref{thm:ddgIsGradeAdjoint} to denote column and row supermatrices. Finally, after the $uosp(1,2;\C\La_N)$ algebra has been defined they  denote normalized even column and row supermatrices we call superqubits.
\end{rem}
The grade adjoint $\overline{ST}$ is not general enough for $\MM$, though. We would like to generalize the double dagger map $\ddg$ for the morphisms of the studied $R$-bimodule $\La^{p|q}$ represented by the supermatrices $S\in\MM$ (this is our starting point in Sec.~\ref{sec:squbits}) and that calls for a generalization of complex conjugation for Grassmann variables. But that again means to  sacrifice the requirement for the form to be Hermitian (let alone positive semidefinite). The way to recover it is the development after Theorem~\ref{thm:ddgIsGradeAdjoint} in the main body of the paper  leading to the $uosp(1|2;\C\La_N)$ algebra (Def.~\ref{def:uosp12}). Now we will present the last missing ingredient to be able to formulate it. Every $\bbZ_2$-graded ring is associated with (at least)  two types of antilinear automorphisms:
\begin{defi}\label{def:automorph}
(i) Let $R$ be a complex supercommutative superalgebra and let there be an automorphism $*:R\mapsto R$ defined as
\begin{subequations}\label{eq:staraction}
\begin{align}
     (ar)^*&=\bar a r^*,\label{eq:staractioni}\\
     (rs)^*&=s^*r^*,\label{eq:staractionii}\\
     (r^*)^*&=r,\label{eq:staractioniii}
\end{align}
\end{subequations}
for all $r,s\in R$ and $a\in\bbC$ where the bar denotes complex conjugation.  \\
(ii) Let the hash map $\#:R\mapsto R$ be defined as
\begin{subequations}\label{eq:hashaction}
\begin{align}
     (ar)^\#&=\bar a r^\#,\label{eq:hashactioni}\\
     (rs)^\#&=(-)^{|r||s|}s^\#r^\#=r^\#s^\#,\label{eq:hashactionii}\\
     (r^\#)^\#&=(-)^{|r|}r.\label{eq:hashactioniii}
\end{align}
\end{subequations}
\end{defi}
\begin{rem}
The star map is an involution and the hash map is a grade involution. Both maps reduce to  ordinary complex conjugation for complex numbers. The star map is frequently used in calculations of fermion path integrals in QFT~\cite{qft} where Grassmann variables appear as well. For us, however, the hash map will be relevant (see Theorem~\ref{thm:ddgIsGradeAdjoint} that would not be possible to formulate with the star involution). For further details consult~\cite{ritt,SchW}. An insight from physics  into the existence of the star and hash maps is provided by Lemma~\ref{lem:conjsfromreality}.
\end{rem}

\section{Calculations with supermatrices}
\label{sec:AppSLinearALg}

We will not list all properties of supermatrices~\cite{buch} but only those few repeatedly used in the main body of the paper. An important map is the supertrace defined for $S\in\MM$ by
\begin{equation}\label{eq:sTr}
  \sTr{}(S)\overset{\rm df}{=}\Tr{}(A)-(-)^{|S|}\Tr{}(D),
\end{equation}
using the standard basis. The following property of the supertrace holds:
\begin{align*}
 \sTr{}(ST) &= (-)^{|S||T|}\sTr{}(T)\sTr{}(S).
\end{align*}
Another important operation is the left and right scalar multiplication of a  supermatrix $S\in\MM$ by a Grassmann number $\zeta\in\C\La_N$ defined as
\begin{align}
 \zeta S &=
\begin{pmatrix}\label{eq:left}
  \zeta A & \zeta B \\
  (-)^{|\zeta|}\zeta C & (-)^{|\zeta|}\zeta D \\
\end{pmatrix}, \\
  S\zeta&=\label{eq:right}
\begin{pmatrix}
   A\zeta &  (-)^{|\zeta|} B \zeta\\
   C\zeta &  (-)^{|\zeta|} D \zeta\\
\end{pmatrix}.
\end{align}
Clearly, the grade of the Grassmann numbers make sense only for homogeneous elements but the linear character of supernumbers from Def.~\ref{def:supernumber} extends its action to an arbitrary supernumber.

The most important consequence of the above rule is that odd Grassmann numbers anticommute with  odd supermatrices. Let $|\zeta|=|S|=1$ let $S$ be written in the standard basis Eq.~(\ref{eq:freebasis}) (as are all supermatrices in this paper). Then Eqs.~(\ref{eq:left}) and (\ref{eq:right}) imply:
\begin{equation}\label{eq:leftANDright}
  \zeta S=\zeta
  \begin{pmatrix}
  A & B \\
  C & D \\
  \end{pmatrix}=
  \begin{pmatrix}
  \zeta A & \zeta B \\
  -\zeta C & -\zeta D \\
  \end{pmatrix}=
  \begin{pmatrix}
   -A\zeta &  B\zeta \\
   -C\zeta &  D\zeta \\
  \end{pmatrix}=
  \begin{pmatrix}
  -A & -B \\
  -C & -D \\
  \end{pmatrix}\zeta=
  -S\zeta.
\end{equation}
Recall that if $S$ is odd then the entries of $A,D$ are odd and of $B,C$ are even. The most important example of the above rule is the following expression:
\begin{equation}\label{eq:zetaANDoddGens}
  \zeta Q_i=-Q_i\zeta,
\end{equation}
where $Q_i$ are the generators from Eqs.~(\ref{eq:uospOddGens}).

Since we agreed that column and row supervectors  are just special cases of supermatrices, the  rule also dictates the behavior of the odd basis state $\kbd$ from Eq.~(\ref{eq:superqubitBases}). Hence, we can write
\begin{equation}\label{eq:zetaVSbulletKet}
  \zeta\kbd=\zeta
  \begin{pmatrix}
    0 \\
    0 \\
    1 \\
  \end{pmatrix}=
    \begin{pmatrix}
    0 \\
    0 \\
    -\zeta \\
  \end{pmatrix}=-
    \begin{pmatrix}
    0 \\
    0 \\
    1 \\
  \end{pmatrix}\zeta=
  -\kbd\zeta
\end{equation}
and similarly for the row vector
\begin{equation}\label{eq:zetaVSbulletBra}
  \zeta\bbd=-\bbd\zeta.
\end{equation}
We will often use
$$
-\zeta\kbd=
    \begin{pmatrix}
    0 \\
    0 \\
    \zeta \\
    \end{pmatrix}
$$
and
$$
\zeta\bbd=
\begin{pmatrix}
  0, & 0, & \zeta \\
\end{pmatrix}.
$$
\begin{exa}
  As an example combining all the salient points of calculations with supermatrices, let's compute Eq.~(\ref{eq:UOSPgroupS}) from the exponential in Eq.~(\ref{eq:UOSPoperator}) by setting $\zeta=2p\eta$ as a result of Lemma~\ref{lem:Ssimplified}:
\begin{subequations}\label{eq:Smatrix}
      \begin{align}
        \exp{[2p\eta Q_1+2p\eta^\# Q_2]} = & \id+\big(2p\eta Q_1+2p\eta^\# Q_2\big)+{1\over2}\big(2p\eta Q_1+2p\eta^\# Q_2\big)\big(2p\eta Q_1+2p\eta^\# Q_2\big) \\
         = &  \id+p\eta
         \begin{pmatrix}
        0 & 0 & 0 \\
        0 & 0 & -1 \\
        -1 & 0 & 0 \\
        \end{pmatrix}
            +p\eta^\#
        \begin{pmatrix}
            0 & 0 & -1 \\
            0 & 0 & 0 \\
            0 & 1 & 0 \\
         \end{pmatrix}\nn\\
         & +2p^2\eta Q_1\eta^\#Q_2+2p^2\eta^\#Q_2\eta Q_1\\
         = & \begin{pmatrix}
            1 & 0 & -p\eta^\# \\
            0 & 1 & -p\eta \\
            p\eta & -p\eta^\# & 1 \\
        \end{pmatrix}
        -
        2p^2\eta\eta^\#(Q_1Q_2-Q_2Q_1)\\
        = & \begin{pmatrix}
            1 & 0 & -p\eta^\# \\
            0 & 1 & -p\eta \\
            p\eta & -p\eta^\# & 1 \\
        \end{pmatrix}
        +
        p^2\eta\eta^\#\begin{pmatrix}
                  {1\over2} & 0 & 0 \\
                  0 & {1\over2} & 0 \\
                  0 & 0 & -1 \\
                \end{pmatrix}\\
        = & \begin{pmatrix}
            1+{p^2\over2}\eta\eta^\# & 0 & -p\eta^\# \\
            0 & 1+{p^2\over2}\eta\eta^\# & -p\eta \\
            p\eta & -p\eta^\# & 1-p^2\eta\eta^\# \\
        \end{pmatrix}.
      \end{align}
    \end{subequations}
    The first equality is all that is left from the Taylor series of the exponential function, in the second equality Eqs.~(\ref{eq:uospOddGens}) were used and the third row comes from Eq.~(\ref{eq:left}), the rule Eq.~(\ref{eq:zetaANDoddGens}) and $\eta\eta^\#=-\eta^\#\eta$.
\end{exa}
\begin{exa}
  Another exercise is the calculation of the norm of Eq.~(\ref{eq:superqubit}) in two different ways: using column/row matrices and kets and bras. Note that the situation is not that straightforward as in ordinary quantum mechanics due to Eqs.~(\ref{eq:zetaVSbulletKet}) and (\ref{eq:zetaVSbulletBra}). From (\ref{eq:superqubit}) and (\ref{eq:superqubitbra}) we get
  \begin{align}\label{eq:NoemrSuperqubit1}
    \brk{\psi}{\psi} & =
    \begin{pmatrix}
     \bar\a \left(1+{p^2\over2}\eta\eta^\#\right) ,& \bar\b \left(1+{p^2\over2}\eta\eta^\#\right), & p(\bar\a\eta^\#+\bar\b\eta)
   \end{pmatrix}
      \begin{pmatrix}
     \a \left(1+{p^2\over2}\eta\eta^\#\right)\\
     \b \left(1+{p^2\over2}\eta\eta^\#\right) \\
     p(\a\eta-\b\eta^\#) \\
   \end{pmatrix}\nn\\
   & =(|\a|^2+|\b|^2)(1+p^2\eta\eta^\#)+p^2(\bar\a\eta^\#+\bar\b\eta)(\a\eta-\b\eta^\#)\nn\\
   & = 1,
  \end{align}
  since $|\a|^2+|\b|^2=1$. On the other hand, one also gets from (\ref{eq:superqubit}) and (\ref{eq:superqubitbra}) the ket/bra version:
  \begin{align}\label{eq:NoemrSuperqubit2}
    \brk{\psi}{\psi} & = |\a|^2+|\b|^2)(1+p^2\eta\eta^\#) - p^2(\bar\a\eta^\#+\bar\b\eta)\bbd(\a\eta-\b\eta^\#)\kbd\nn\\
    & = (1+p^2\eta\eta^\#)+p^2(\bar\a\eta^\#+\bar\b\eta)(\a\eta-\b\eta^\#)\nn\\
    & = 1,
  \end{align}
  where in the second equality we used (\ref{eq:zetaVSbulletBra}).
\end{exa}
\begin{exa}
  It is instructive to see how the multiplication rules in Eq.~(\ref{eq:leftANDright}) are compatible with the supertranspose $ST$. In Eq.~(\ref{eq:uosparbitraryelement}) we used $(\eta Q_1)^\ddg=-\eta^\# Q_2$ valid for $\eta\in\C\La_{N,1}$ and $ST$ defined in Eq.~(\ref{eq:STforS}) where $\ddg=\#\circ ST\equiv ST\circ\#$ (see Theorem~\ref{thm:ddgIsGradeAdjoint}). Let's verify it by a (different) direct calculation:
  \begin{align}\label{eq:etaANDQ1direct}
    (\eta Q_1)^\ddg&={1\over2}
    \begin{pmatrix}
        0 & 0 & 0 \\
        0 & 0 & -\eta \\
        \eta & 0 & 0 \\
    \end{pmatrix}^\ddg={1\over2}
    \begin{pmatrix}
        0 & 0 & 0 \\
        0 & 0 & -\eta^\# \\
        \eta^\# & 0 & 0 \\
    \end{pmatrix}^{ST}={1\over2}
    \begin{pmatrix}
        0 & 0 & \eta^\# \\
        0 & 0 & 0 \\
        0 & \eta^\# & 0 \\
    \end{pmatrix}\nn\\
    &=\eta^\#{1\over2}
    \begin{pmatrix}
        0 & 0 & 1 \\
        0 & 0 & 0 \\
        0 & -1 & 0 \\
    \end{pmatrix}=
    -\eta^\#Q_2.
  \end{align}
  In the first equality we used~Eq.~(\ref{eq:leftANDright}), the second equality is the definition of $\ddg$, in the third equality the supertranspose Eq.~(\ref{eq:STforS}) was applied (note that the matrix becomes even after $\eta$ has `entered' $Q_1$) and the in the fourth equality  Eq.~(\ref{eq:leftANDright}) was used again.
\end{exa}

\begin{acknowledgements}
The author acknowledges support from the Office of Naval Research (grant No. N000140811249) and The Royal Society (International Exchanges travel grant) and is grateful to  Michael Duff, Leron Borsten, Markus M\"uller and Daniel Gottesmann for  comments and  G\'abor Luk\'acs for discussions.
\end{acknowledgements}


\begin{thebibliography}{10}

\bibitem{bell}J. S. Bell, Physics {\bf1}, 195 (1964).
\bibitem{tsirelson}B. S. Tsirelson, Letters in Mathematical Physics {\bf4}, 93 (1980).
\bibitem{CHSH}J. F. Clauser, M. A. Horne, A. Shimony and R. A. Holt, Physical Review Letters {\bf23}, 880 (1969).
\bibitem{PR}S. Popescu and D. Rohrlich, Foundations of Physics {\bf24}, 379 (1994).
\bibitem{CHSHgame}R. Cleve, P. H\o yer, B. Toner and J. Watrous, Proceedings of the $19^{th}$ IEEE Annual Conference on Computational Complexity, (2004). arXiv:quant-ph/0404076.
\bibitem{complexity}W. van Dam, arXiv:quant-ph/0501159. G. Brassard, H. Buhrman, N. Linden, A. M\'ethot, A. Tapp and F. Unger, Physical Review Letters {\bf96}, 250401 (2006).
\bibitem{ent}L. Masanes, M. M\"uller, D. P\'erez-Garc\'ia and R. Augusiak, arXiv:1111.4060.
\bibitem{complexity_review}H. Buhrman, R. Cleve, S. Massar and R. de Wolf, Review of Modern Physics {\bf 82}, 665 (2010).
\bibitem{masanes}L. Masanes, A. Ac\'in, and N. Gisin, Physical Review A {\bf73}, 012112 (2006).
\bibitem{IC}M. Paw\l owski, T. Paterek, D. Kaszlikowski, V. Scarani, A. Winter  and M. \.Zukowski, Nature {\bf461}, 1101 (2009).
\bibitem{dataproc}O. C. O. Dahlsten, D. Lercher and R. Renner, arXiv:1108.4549.
\bibitem{MI1}S. W. Al-Safi and A. J. Short, Physical Review A {\bf84}, 042323 (2011).
\bibitem{MI2}E. Wakakuwa and M. Murao, arXiv:1207.2286.
\bibitem{qubit}R. Mosseri and R. Dandoloff, Journal of Physics A: Mathematical and General {\bf34},10243 (2001).

\bibitem{classic}S. MacLane, G. Birkhoff, Algebra ($3^{rd}$ ed., AMS Chelsea Publishing, 1999).
\bibitem{superqubits}L. Borsten, D. Dahanayake, M. J. Duff and W. Rubens, Physical Review D {\bf81}, 105023 (2010).
\bibitem{bbd}L. Borsten, K. Br\'adler and M. J. Duff, arXiv:1206.6934.
\bibitem{ritt}V. Rittenberg, Lecture Notes in Physics {\bf79}, 3 (Springer, Heidelberg, 1978).
\bibitem{manin}Y. I. Manin, Gauge Field Theory and Complex Geometry ($2^{nd}$ ed., Springer, 2002).
\bibitem{superLie}V. G. Kac, Advances in Mathematics {\bf26}, 8 (1977).
\bibitem{superLie1}M. Scheunert, The Theory of Lie Superalgebras, An Introduction (Lecture Notes in Mathematics, Springer, 1979).
\bibitem{varadara}V. S. Varadarajan, Supersymmetry for Mathematicians: An Introduction (AMS, 2004).
\bibitem{berezin_book}F. A. Berezin, Introduction to Superanalysis (Kluwer, 1987).
\bibitem{kostant}B. Kostant, Graded Manifolds, Graded Lie Theory, and Prequantization, in K. Bleuler and A. Reetz, eds.,  Differential Geometric Methods in Mathematical Physics~I, Lecture Notes in Mathematics~{\bf570} (Springer, Berlin, 1977).
\bibitem{carmeli}C. Carmeli, L. Caston and R. Fioresi, Mathematical Foundations of Supersymmetry (European Mathematical Society, 2011).
\bibitem{buch}I. L. Buchbinder and S. M. Kuzenko, Ideas And Methods of Supersymmetry and Supergravity, Or, a Walk Through Superspace (Taylor \& Francis, 1998).

\bibitem{snr2}M. Scheunert, W. Nahm and V. Rittenberg, Journal of Mathematical Physics  {\bf18}, 146 (1977).
\bibitem{snr1}M. Scheunert, W. Nahm and V. Rittenberg, Journal of Mathematical Physics {\bf18}, 155 (1977).
\bibitem{bch}V. A. Kosteleck\'y, M. M.  Nieto and R. Truax, Journal of Mathematical Physics  {\bf27}, 1419 (1986).
\bibitem{berezin}F. A. Berezin and V. N. Tolstoy, Communications in Mathematical Physics {\bf78}, 409 (1981).
\bibitem{rittscheu}V. Rittenberg and M. Scheunert, Journal of Mathematical Physics  {\bf19}, 709 (1978).
\bibitem{landimarmo}G. Landi and G. Marmo, Physics Letters B {\bf193}, 61 (1987).
\bibitem{cohstates}M. Chaichian, D. Ellinas and P. Pre\v snajder, Journal of Mathematical Physics {\bf32}, 3381 (1991).
\bibitem{supersphere}H. Grosse, C. Klim\v c\'ik and P. Pre\v snajder, Communications in Mathematical Physics {\bf185}, 155 (1997).
\bibitem{supersphere1}C. Bartocci, U. Bruzzo and G. Landi, Journal of Mathematical Physics {\bf31}, 45 (1990).
\bibitem{SchW}A. F. Schunck and C. Wainwright, Journal of Mathematical Physics {\bf46}, 033511 (2005).
\bibitem{efetov}K. Efetov, Supersymmetry in Disorder and Chaos (CUP, 1999).
\bibitem{superspinors}K. Hasebe and Y. Kimura, Nuclear Physics~B {\bf709}, 94 (2005).
\bibitem{hasebe}K. Hasebe, K. Totsuka, 	Symmetry {\bf5}, 119 (2013).


\bibitem{nieto}M. M. Nieto, ``Physical interpretation of supercoherent states and their associated Grassmann numbers". Talk presented at the Meeting on the Foundation of Quantum Mechanics, Santa Fe, NM, (May 1991).
\bibitem{rogersnorm}A. Rogers, Journal of Mathematical Physics {\bf21}, 1352 (1980).
\bibitem{supermani}C. Bartocci and U. Bruzzo and D. Hern{\'a}ndez-Ruip{\'e}rez, The Geometry of Supermanifolds (Kluwer Academic Publishers, 1999).
\bibitem{rabincrane}J. M. Rabin and L. Crane, Communications in Mathematical Physics {\bf100}, 141 (1985).
\bibitem{cookfulp}J. Cook and R. Fulp, Differential Geometry and its Applications {\bf26}, 463 (2008).
\bibitem{rudolph}O. Rudolph, Communications in Mathematical Physics {\bf214}, 449 (2000).
\bibitem{castell}L. Castellani, P. A. Grassi and L. Sommovigo, arXiv:1001.3753.

\bibitem{qft}M. Nakahara, Geometry, Topology and Physics (2nd ed, IOP publishing, Bristol, 2003).
\bibitem{qft1}G. B. Folland, Quantum Field Theory: A Tourist Guide For Mathematicians (AMS, 2008).
\bibitem{qft2}G. Sterman, An Introduction to Quantum Field Theory (CUP, 1993).
\bibitem{qft3}M. A. Srednicki, An Introduction to Quantum Field Theory (CUP, 2007).
\bibitem{musson}I. M. Musson, Lie superalgebras and enveloping algebras (AMS, 2012).

\bibitem{fulton}W. Fulton, Algebraic Topology (Springer, 1995).

\bibitem{yalmip}YALMIP: A Toolbox for Modeling and Optimization in MATLAB. J. L\"ofberg. In Proceedings of the CACSD Conference, Taipei, Taiwan, 2004.

\bibitem{found1}H. Barnum, J. Barrett, M. Leifer and A. Wilce, Physical Review Letters {\bf99}, 240501 (2007).
\bibitem{found2}L. Hardy, arXiv:quant-ph/0101012.
\bibitem{found3}M. M\"uller and C. Ududec, Physical Review Letters {\bf108}, 130401 (2012).

\end{thebibliography}
\end{document}